\documentclass[a4paper,11pt]{article}
\usepackage[pdftex]{graphicx}
\usepackage{fullpage}
\usepackage{amsmath,amsthm,amssymb,amsfonts}
\usepackage{mathtools}
\usepackage{bm}
\allowdisplaybreaks

\usepackage{hyperref}
\usepackage[capitalise]{cleveref}
\usepackage{thm-restate}

\usepackage{comment}
\usepackage{enumerate}

\newtheorem{definition}{Definition}[section]
\newtheorem{lemma}[definition]{Lemma}
\newtheorem{theorem}[definition]{Theorem}
\newtheorem{corollary}[definition]{Corollary}
\newtheorem{proposition}[definition]{Proposition}
\newtheorem{note}[definition]{Note}
\newtheorem{claim}[definition]{Claim}
\newtheorem{remark}[definition]{Remark}

\title{\textbf{Phase Transition of the \texorpdfstring{$k$}{k}-Majority Dynamics\\in Biased Communication Models}}
\author{Emilio Cruciani\\
		{\small{}Paris-Lodron University of Salzburg}\\
		{\small{}Salzburg, Austria}\\
		{\small{}\texttt{emilio.cruciani@plus.ac.at}}
	\and \hspace{11ex}Hlafo Alfie Mimun\\
		{\small{}\hspace{14ex}LUISS Guido Carli}\\
		{\small{}\hspace{14ex}Rome, Italy}\\
		{\small{}\hspace{14ex}\texttt{hmimun@luiss.it}}
    \and \hspace{5ex}Matteo Quattropani\\
		{\small{}\hspace{6ex}Sapienza University of Rome}\\
		{\small{}\hspace{6ex}Rome, Italy}\\
		{\small{}\hspace{6ex}\texttt{matteo.quattropani@uniroma1.it}} 
    \and \hspace{11ex}Sara Rizzo\\
		{\small{}\hspace{13ex}Gran Sasso Science Institute}\\
		{\small{}\hspace{13ex}L'Aquila, Italy}\\
		{\small{}\hspace{13ex}\texttt{sara.rizzo@gssi.it}}
}
\date{}

\crefname{claim}{Claim}{Claims}
\usepackage[framemethod=TikZ]{mdframed}
\mdfsetup{%
    backgroundcolor=gray!10,
    leftmargin=10pt,
    rightmargin=10pt,
    middlelinewidth=0pt,
    roundcorner=5pt,
}
\usepackage{dsfont}
\usepackage{xspace}

\newcommand{\keywords}[1]{\small \textbf{Keywords: }#1}

\newcommand{\robust}{robust\xspace}
\newcommand{\robustness}{robustness\xspace}

\newcommand{\bigO}{\mathcal{O}\xspace}
\newcommand{\card}[1]{\lvert #1 \rvert}

\newcommand{\dfn}{\coloneqq\xspace}  
\newcommand{\neigh}[1]{N_{#1}}
\renewcommand{\deg}[1]{\delta_{#1}}

\newcommand{\voter}{\textsc{Voter Model}\xspace}
\newcommand{\choices}{2-\textsc{Choices}\xspace}
\newcommand{\undecided}{\textsc{Undecided State}\xspace}
\newcommand{\maj}{\textsc{Majority}\xspace}
\newcommand{\determ}{\textsc{Deterministic}\xspace}
\newcommand{\nodedeterm}{\textsc{Node deterministic}\xspace}
\newcommand{\majority}[3]{$(#1, #2, #3)$-\textsc{Edge} \maj{}}
\newcommand{\modmajority}[3]{$(#1, #2, #3)$-\textsc{Node} \maj{}}
\newcommand{\kmaj}{$k$-\maj{}\xspace}

\newcommand{\detmaj}[2]{$(#1,#2)$-\textsc{edge} \determ{} \maj{}\xspace}
\newcommand{\moddetmaj}[2]{$(#1,#2)$-\nodedeterm{} \maj{}\xspace}

\newcommand{\red}{\ensuremath{\mathcal{R}}\xspace}
\newcommand{\blue}{\ensuremath{\mathcal{B}}\xspace}
\newcommand{\state}[2]{x_{#1}^{(#2)}}
\newcommand{\conf}[1]{\mathbf{x}^{(#1)}}
\newcommand{\bstate}[3]{\bar{x}_{#2}^{(#3)}(#1)}

\newcommand{\sample}[2]{S_{#1}^{(#2)}}
\newcommand{\R}[1]{R^{(#1)}}
\newcommand{\B}[1]{B^{(#1)}}
\newcommand{\bR}[2]{\bar{R}_{#1}^{(#2)}}
\newcommand{\bB}[2]{\bar{B}_{#1}^{(#2)}}
\newcommand{\Fpk}[2]{F_{#1,#2}}

\newcommand{\vroot}{v_0}

\newcommand{\rfrac}[2]{\phi_{#1}^{(#2)}}
\newcommand{\rfracmax}[1]{\phi_{_{\max}}^{(#1)}}

\newcommand{\Emaj}{\mathcal{E}}

\newcommand{\Prob}[1]{\mathbf{P}\left(#1\right)}
\newcommand{\Probb}[2]{\mathbf{P}_{#1}\left(#2\right)}
\newcommand{\Ex}[1]{\mathbf{E}\left[#1\right]}
\newcommand{\cond}{\ \middle| \ }
\newcommand{\condconf}{ \cond \conf{t} }
\newcommand{\bin}[2]{{\rm Bin}\left( #1, #2 \right)}
\newcommand{\event}[2]{\mathcal{E}_{#1}^{(#2)}}

\newcommand{\pk}{p_k^\star}
\newcommand{\pkq}{p_{k,q}^\star}
\newcommand{\modpkq}{\hat{p}_{k,q}^\star}

\begin{document}
\maketitle

\begin{abstract}
Consider a graph where each of the $n$ nodes is either in state \red or \blue. Herein, we analyze the \emph{synchronous $k$-\maj{} dynamics}, where in each discrete-time round nodes simultaneously sample $k$ neighbors uniformly at random with replacement and adopt the majority state among those of the nodes in the sample (breaking ties uniformly at random).

Differently from previous work, we study the \robustness of the $k$-\maj{} in \emph{maintaining a \red majority}, when the dynamics is subject to two forms of \emph{bias} toward state \blue.
The bias models an external agent that attempts to subvert the initial majority by altering the communication between nodes, with a probability of success $p$ in each round:
in the first form of bias, the agent tries to alter the communication links by transmitting state \blue;
in the second form of bias, the agent tries to corrupt nodes directly by making them update to \blue.

Our main result shows a \emph{sharp phase transition} in both forms of bias.
By considering initial configurations in which every node has probability $q \in (\frac{1}{2},1]$ of being in state \red, we prove that for every $k\geq3$ there exists a critical value $\pkq$ such that, with high probability, the external agent is able to subvert the initial majority either in $n^{\omega(1)}$ rounds, if $p<\pkq$, or in $O(1)$ rounds, if $p>\pkq$.
When $k<3$, instead, no phase transition phenomenon is observed and the disruption happens in $O(1)$ rounds for $p>0$.
\end{abstract}

\vfill
\noindent
\textbf{Acknowledgments:} H.\ A.\ Mimun and M.\ Quattropani are members of GNAMPA-INdAM and they acknowledge partial support by the GNAMPA-INdAM Project 2020 ``Random walks on random games'' and PRIN 2017 project ALGADIMAR.
E.\ Cruciani: Supported by the Austrian Science Fund (FWF): P 32863-N; this project has received funding from the European Research Council (ERC) under the European Union's Horizon 2020 research and innovation programme (grant agreement No 947702).

\newpage
\section{Introduction}\label{sec:intro}

Designing distributed algorithms that let the nodes of a graph reach a \emph{consensus}, i.e., a configuration of states where all the nodes agree on the same state, is a fundamental problem in distributed computing and multi-agent systems. 
Consensus algorithms are used, for example, for leader election, atomic broadcast, and clock synchronization problems~\cite{DBLP:journals/dc/BoczkowskiKN19, DBLP:journals/dc/FeinermanHK17}.
Recently there has been a growing interest in the analysis of \emph{dynamics} as distributed algorithms for the consensus problem~\cite{DBLP:journals/dc/BecchettiCNPST17,DBLP:conf/mfcs/ClementiGGNPS18,DBLP:conf/icalp/CooperER14,DBLP:conf/spaa/DoerrGMSS11,DBLP:conf/podc/GhaffariL18,hassin2001voter,DBLP:journals/aamas/MosselNT14}, inspired by simple mechanisms studied in statistical mechanics for interacting particle systems~\cite{liggett2012interacting}.
In this scenario, nodes are \emph{anonymous} (i.e., they do not have distinct IDs) and they have a state that evolves over time according to some simple common local rule based on the states of their neighbors.
We discuss the literature on dynamics for consensus in \cref{sec:related}.

Herein we focus on the scenario in which every node has a \emph{binary state} (either \red or \blue) and the communication (defined by an underlying graph, not necessarily connected and potentially directed) proceeds in \emph{synchronous} rounds. 
We only assume the communication graph to be \emph{sufficiently dense}, namely the minimum (out-)degree must be superlogarithmic in the number of nodes.
In this setting we analyze the $k$-\maj{} dynamics, where nodes update their state to that held by the majority of a random sample of $k$ neighbors. 
This class of dynamics generalizes other well-known dynamics, e.g., $1$-\maj{} a.k.a.\ \voter{}~\cite{hassin2001voter} and $3$-\maj{}~\cite{DBLP:journals/dc/BecchettiCNPST17}.

We analyze a novel scenario that allows to study the \robustness of the dynamics in \emph{maintaining a majority} when the dynamics is subject to two different forms of adversarial noise, that we call \textit{bias}.
We assume an initial configuration where each node is in state \red, independently of the others, with probability $q>\frac{1}{2}$; this implies an initial majority on state \red \emph{with high probability} (i.e., with probability $1-o(1)$).
The two forms of bias model an external adversarial agent that at each round tries to subvert the initial majority with a probability of success $p$.
In the first form of bias, the agent is able to alter the communication between pairs of nodes (as in a man-in-the-middle attack~\cite{perlman2016network}); we model such a scenario through the use of Z-channels (binary asymmetric channels studied in information theory and used to model the effect of noise on communication~\cite{mackay2003information}):
state \blue is always transmitted correctly, while state \red is transmitted incorrectly with probability $p$.
In the second form of bias, the agent can directly corrupt nodes: when the agent succeeds, with probability $p$, nodes change state to \blue independently of the states of their neighbors.
Note that the first model of bias has been introduced in~\cite{DBLP:conf/atal/CrucianiNNS18}, where the authors use it as a tool to analyze the behavior of a dynamics on core-periphery networks.
The second model of bias, instead, has been introduced in~\cite{DBLP:conf/ijcai/Anagnostopoulos20}, in an asynchronous communication setting, in order to model the diffusion of a superior opinion in networks.
However, both the applications we consider and the mathematical framework we build to analyze the models under a common lens are novel contributions of this paper.
A more detailed description of such dynamics and of other existing models of bias is deferred to \cref{sec:related}.

When the \kmaj is subject to such forms of bias, we respectively call it \majority{k}{p}{\blue} (\cref{def:majority}) and \majority{k}{p}{\blue} (\cref{def:modmajority}).
Technically, when $p>0$, both dynamics are Markov Chains with a single absorbing state in which all nodes are in state \blue{}.
This implies that, since the graphs are finite, such an absorbing state will be reached in a finite number of rounds with probability~1, i.e., the bias will eventually subvert the initial majority.
For this reason, we study how the \emph{time of disruption} $\tau$, that is the first round in which the initial \red majority is subverted (\cref{def:time_disruption}), depends on $k$, $p$, and $q$.
Trivially, for both biased models, if $p=0$ the process remains stuck in its initial configuration, while if $p=1$ the process reaches the absorbing state in a single round.
More generally, it is intuitively clear that the process will converge slowly to the absorbing state if $p$ is small and more quickly if $p$ is large.
However, the behavior is nonlinear in $p$ and indeed, with such an intuition in mind, we prove a \emph{sharp phase transition} between \emph{slow} and \emph{fast disruption regimes}.

\begin{mdframed}
    \textbf{Informal Description of \cref{th:phasetransition-initial}.}
	Consider a graph of $n$ nodes such that the minimum degree is $\omega(\log n)$.
	Let every node initially be in state \red{} with probability $q>\frac{1}{2}$, independently of the others.
	For every constant $k \geq 3$, there exists a constant $\pkq$ (resp.\ $\modpkq$) such that, for the \majority{k}{p}{\blue} (resp.\ \modmajority{k}{p}{\blue}) with bias $p$:
	\begin{itemize}
	\item \emph{Slow disruption:} If $p < \pkq$ (resp.\ $p<\modpkq$), then $\tau = n^{\omega(1)}$ with high probability;
	\item \emph{Fast disruption:} If $p > \pkq$ (resp.\ $p>\modpkq$), then $\tau=O(1)$ with high probability.
	\end{itemize}
\end{mdframed}

\cref{th:phasetransition-initial} also states that, in the slow disruption regime, after a constant number of rounds the fraction of \red nodes remains concentrated around a constant value greater than $\frac{1}{2}$; for this reason we also call such a regime \emph{metastable}.
From the previous description, the two dynamics appear identical. However their behavior is only \emph{qualitatively} the same, given that they can differ in the critical values and in the metastable fraction of nodes in state \red.
In fact, the critical values $\pkq$ and $\modpkq$ do not always coincide, as described in detail in \cref{rmk:pkq compare} and depicted (for $k=3$) in the phase diagrams in \cref{fig:phasediagrams}.

In particular, the critical values are equal to a value $\pk$ whenever the initial \red density $q$ is larger than a threshold $\varphi_{\pk,k}$, for \majority{k}{p}{\blue}, or $\hat\varphi_{\pk,k}$, for \modmajority{k}{p}{\blue}; this happens, e.g., when starting from an initial \red consensus.
When $q$ is smaller than such thresholds, instead, the critical values of the two dynamics become different.
In this sense, when $q$ is small, the effect of the bias in the \majority{k}{p}{\blue} is stronger than in the \modmajority{k}{p}{\blue}; in fact, whenever the initial \red density $q$ is below $\varphi_{\pk,k}$, there exist values of the bias $p$ such that the former dynamics is in the fast disruption regime while the latter is in the slow disruption regime.
Formally, it holds that $\modpkq \ge \pkq$.
Such phenomenon can also be observed in \cref{fig:phasediagrams} by noting that the red area in the right diagram, representing the slow disruption regime, is larger than the red area in the left diagram.

However, as precisely stated in \cref{th:phasetransition-initial}, the \robustness of the $k$-\maj{} in the two models of bias can be also measured as the fraction of nodes in state \red during the metastable phase, which we respectively call $\varphi_{p,k}^+$ and $\hat\varphi_{p,k}^+$.
From this other point of view, the effect of bias in the \majority{k}{p}{\blue} is weaker than in the \modmajority{k}{p}{\blue}.
In fact, in the slow disruption regime, the concentration of \red nodes in the former dynamics is higher than in the latter, namely $(1-p)\varphi_{p,k}^+=\hat\varphi_{p,k}^+$ (as follows from \cref{lem:fixed_points_F-prel,lem:fixed_points_F-prel2} and from \cref{eq:equivalence}).
This phenomenon can be also observed in \cref{fig:phasediagrams} by noting that the red area on the left diagram is more intense than that on the right diagram.
We discuss applications of such a metastable behavior in \cref{sec:applications}.

\begin{figure}[ht]
\begin{center}\hfill%
    \begin{minipage}[t]{0.5\textwidth}
      \includegraphics[width=\linewidth]{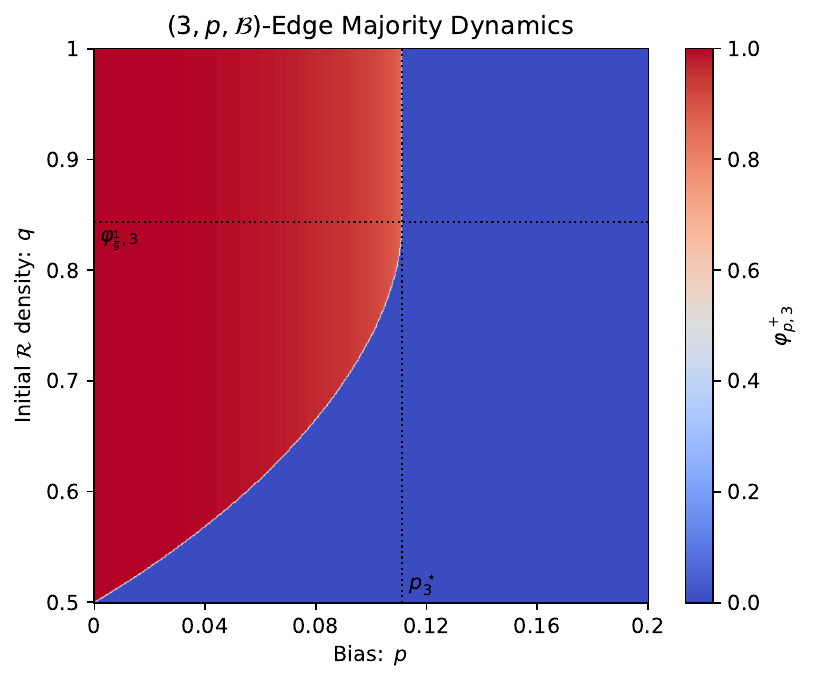}
    \end{minipage}\hfill%
    \begin{minipage}[t]{0.5\textwidth}
      \includegraphics[width=\linewidth]{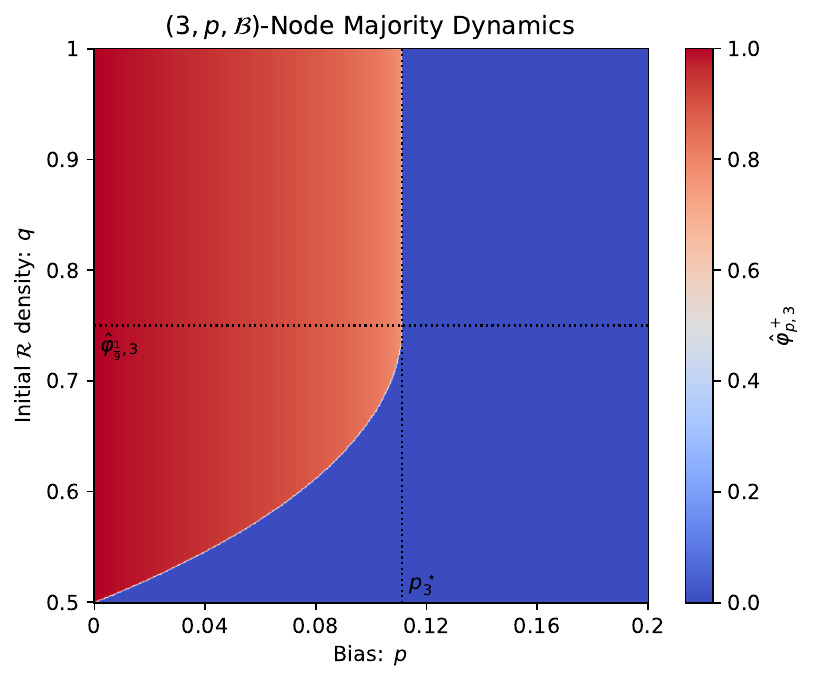}
    \end{minipage}\hfill%
\end{center}
\caption{\footnotesize
The diagrams show where the density of \red nodes (here denoted by $\varphi^+_{p,3}$ and $\hat\varphi^+_{p,3}$) concentrates, with high probability, in every round $t \in [T, n^K]$ for a sufficiently large positive constant $T$ and for any positive constant $K$, with different choices of the bias $p$ ($x$-axis) and the initial $\red$ density $q$ ($y$-axis), for the \majority{3}{p}{\blue} (left) and \modmajority{3}{p}{\blue} (right).
The diagrams also show the critical values $p_3^\star=\frac{1}{9}$, $\varphi_{\frac{1}{9},3}=0.84375$, and $\hat\varphi_{\frac{1}{9},3}=0.75$ (dotted lines).}
\label{fig:phasediagrams}
\end{figure}

We further analyze the ``limit cases'' of the $k$-\maj{} in \cref{sec:limit}.
When $k<3$ the dynamics reduces to the \voter{}; we prove that it does not exhibit any phase transition phenomenon and that, in particular, it is not \robust to the external bias, resulting in a constant time of disruption for every arbitrarily small constant $p$ (see \cref{thm:voter}).
When instead the sample is large, either in \determ{} \maj{} (the sample of a node deterministically coincides with its neighborhood, \cref{def:det-majority}) or when the size of the sample $k$ grows with $n$, the dynamics exhibit again sharp phase transitions on the critical value $p^\star = \frac{1}{2}$ (see \cref{prop:det-majority} and \cref{prop:kmaj-klarge}).

Last but not least, we mention here also our analysis of the \emph{tree-processes}, which let us analyze the two dynamics in a simplified topological setting where the graph is a directed infinite tree. 
They are presented before the other results, in \cref{sec:preli}, since they allow to grasp the intuitions behind the proofs without getting involved in technical details, allowing to remove some dependencies in the process. 
Moreover, they allow to easily understand the ``duality'' of the two models of bias, i.e., the fact that the concentration of \red in the metastable phase (slow disruption regime) is related by a linear relation (\cref{eq:equivalence}).
In more detail, the phase transition in these models concerns the probability of the root to be in state \blue as the number of rounds goes to infinity (\cref{th:tree,th:tree2}).
Roughly speaking, the tree-processes look at the local point of view of a node $v$, showing how its state at round $t$ solely depends on the states of the nodes $v$ sampled in round $t-1$; how, in turn, the states of such nodes at round $t-1$ solely depend on the states of the nodes they sampled at round $t-2$; and so on until we get the dependency on the initial configuration at round $t=0$.

We remark that this paper extends preliminary work already appeared in conference proceedings~\cite{DISC2020,ICDCN2021} where most of the technical details are omitted for lack of space. 
In the current version all the full proofs are included.
The results are now framed in a different scenario and presented under a different point of view that allows to focus on the \robustness of the $k$-\maj{} dynamics w.r.t.\ adversarial noise, better fitting our models;
the results are also extended to a new form of bias (namely \modmajority{k}{p}{\blue}) and described through a unified mathematical framework.

\subsection{Further Applications and Research Directions}\label{sec:applications}
Other than directly studying the \robustness to adversarial noise of the $k$-\maj{} in the two forms of bias we introduced, our framework has other applications which we discuss in this section.

\subsubsection{Metastability for Community Detection}

The same biased communication model of the \majority{k}{p}{\blue} has been previously used to analyze the \choices{} on core-periphery networks~\cite{DBLP:conf/atal/CrucianiNNS18}.
The main theorems in this paper are more general (we consider $k$-\maj{} for any $k$ and \determ{} \maj{} instead of only the \choices{} dynamics), more comprehensive (we consider a random initialization instead of only a monochromatic one), and more precise in the description of the behavior of the dynamics (e.g., we explicitly derive the constant time in which the metastable phase is reached and quantify the concentration of nodes in state \red in the metastable phase).
Thanks to \cref{th:phasetransition-initial}, it is easy to extend the results of~\cite{DBLP:conf/atal/CrucianiNNS18} also for the $k$-\maj{} dynamics when $k\ge3$.

The arise of a metastable phase makes the framework suitable to design distributed algorithms for community detection based on the $k$-\maj{} dynamics, similarly to what have been done in~\cite{DBLP:conf/aaai/CrucianiNS19}.
In particular, this is possible for a class of graphs known as \emph{volume-regular graphs}, recently introduced in~\cite{becchetti2019sbs} and strictly related to \emph{ordinary lumpable} Markov Chains~\cite{kemeny1960finite}.
Notable examples of volume-regular graphs are those sampled from the \emph{regular} Stochastic Block Model~\cite{holland1983stochastic}, where the nodes are partitioned into several clusters.
As motivating example, consider a volume-regular graph $G=(V, E)$ where the vertex set $V$ is partitioned into two clusters $V_1,V_2$. 
Since $G$ is volume-regular, it has the property that, for every pair of nodes in each cluster, their fraction of neighborhood toward the other cluster equals some constant $z$.
Let $G$ run the $k$-\maj{} and suppose to start from an initial configuration such that the two clusters have a local majority on opposite states, say nodes in $V_1$ agree on \red{} and nodes in $V_2$ on \blue.
Such an initial configuration (in which clusters have opposite majorities) can be obtained, e.g., via a ``lucky'' random initialization in which each node chooses its state between \red and \blue with probability $\frac{1}{2}$ and independently of the others, as shown in~\cite{DBLP:conf/aaai/CrucianiNS19}.
The local evolution of such a process inside, e.g., $V_1$ can be described by the \majority{k}{z}{\blue} dynamics run by the subgraph induced by $V_1$: the effect of noise, with probability $p=z$, mimics the fact that a node in $V_1$ is sampling a neighbor in $V_2$ in the worst-case scenario in which nodes in $V_2$ never change color.
If $G$ is such that $z < \pk$, then \cref{th:phasetransition-initial} implies that the majority in $V_1$ is maintained, w.h.p., for any polynomial number of rounds.
Since the same reasoning, by symmetry, can be done for $V_2$, it follows that the graph would stay, for any polynomial number of rounds and w.h.p., in a configuration that highlights its clustered structure for $n^{\omega(1)}$ rounds. 

Running the dynamics for a suitable number of (parallel) times, as shown in~\cite{DBLP:conf/aaai/CrucianiNS19,becchetti2019sbs},
makes the $k$-\maj{} suitable for the design of distributed community detection protocols.
The results in \cref{sec:limit}, interestingly, show that the \voter{} cannot be used for this task because it does not exhibit a threshold behavior and especially does not exhibit any metastable phase.

\subsubsection{Biased Opinion Dynamics}

The \modmajority{k}{p}{\blue} can also be seen as a biased opinion dynamics used to model the diffusion of an inherently ``superior opinion'' (\blue in our case), e.g., modeling the introduction in the market of a new technology that will eventually substitute the status quo. 
The metastable behavior of the dynamics demonstrates how the density of the graph negatively affects the time needed by the superior opinion to prevail in the network.

Our model has been introduced for this very goal in~\cite{DBLP:conf/ijcai/Anagnostopoulos20}, with the only difference related to the synchronicity of the models:
while in this paper all nodes update their state simultaneously at each round, in~\cite{DBLP:conf/ijcai/Anagnostopoulos20} one node is selected uniformly at random at each time step to update its state.
The bias, however, directly apply to the states of the nodes with a fixed probability of changing them to the superior opinion.
The dynamics considered in~\cite{DBLP:conf/ijcai/Anagnostopoulos20} are the \determ{} \maj{} and the \voter{};
other than those two, herein we also analyze the $k$-\maj{} dynamics.

The results in~\cite{DBLP:conf/ijcai/Anagnostopoulos20} show a sharp phase transition around $p=\frac{1}{2}$ for sufficiently dense graphs and while considering the \determ{} \maj{} dynamics: whenever the bias $p<\frac{1}{2}$ they prove a lower bound which is exponential in the minimum degree (hence, superpolynomial whenever the minimum degree is superlogarithmic);
when $p>\frac{1}{2}$, instead, the convergence to the superior opinion happens quickly in $O(n \log n)$ updates.
The results drastically change while changing the underlying dynamics to the \voter{}: there is no phase transition phenomenon and the convergence to the superior opinion happens quickly in $O(n \log n)$ updates.

In \cref{sec:limit} we analyze the same dynamics in the synchronous model, where at each time step all nodes update simultaneously.
Looking at our results (\cref{thm:voter,prop:moddet-majority}), we note that they essentially match with those in~\cite{DBLP:conf/ijcai/Anagnostopoulos20}. 
When considering \determ{} \maj{}, we observe the same phase transition phenomenon on the critical value $p=\frac{1}{2}$ (\cref{prop:moddet-majority}).
Regarding the slow disruption regime, while considering sufficiently dense graphs, we both get superpolynomial lower bounds in the time of disruption (or convergence of the superior opinion).
Regarding the fast disruption regime, instead, we get a disruption in 2 rounds w.h.p., while in~\cite{DBLP:conf/ijcai/Anagnostopoulos20} $O(n \log n)$ updates are needed. Note that a factor $n$ is due to the different synchronicity of the models, therefore the behavior of the \determ{} \maj{} in the two models match up to a logarithmic factor.
Moreover, both in \cref{thm:voter} and in~\cite{DBLP:conf/ijcai/Anagnostopoulos20} no phase transition is observed when considering the \voter{} as underlying dynamics.

Herein, though, the results are also generalized for the $k$-\maj{}. When considering an initial consensus configuration (i.e., $q=1$), we prove the existence of a threshold value that increases with $k$ and moves from $1/9$ (for $k=3$) up to $1/2$ (for $k=\omega(1)$).
The result presented in this paper could likely be proven also in the asynchronous setting studied in~\cite{DBLP:conf/ijcai/Anagnostopoulos20} by means of different probabilistic techniques, with the same threshold values.

\subsubsection{Further Directions}

It remains unclear whether it is possible to prove similar results for sparser topologies. More precisely, it would be interesting to see how our results could be sharpened by assuming a particular topology as, e.g., an Erd\H{o}s--R\'{e}nyi random graph $G(n,p)$ with $p=\frac{c \log n}{n}$, for a constant $c$ close to the connectivity threshold.

From the statistical physics perspective, it would also be of interest the analysis of the critical case $p=\pk$ on some topologies. In particular, it would be notable if precise asymptotics on the convergence time could be obtained in the critical regime without any topological assumption on the graph.  

Finally, possible research directions that could lead to non-obvious conclusions are that of applying our biased framework to other dynamics, such as the \undecided{} dynamics~\cite{doi:10.1137/1.9781611977073.135} or the $q$-\voter{}~\cite{castellano2009nonlinear}, or to consider more than two states.

\section{Related Work}\label{sec:related}

\subsection{Dynamics for Consensus}
Simple mathematical model of interaction between nodes in a network have been studied since the first half of the 20th century in statistical mechanics, e.g., to model interacting particle systems or ferromagnetism phenomena~\cite{liggett2012interacting}.
The simplest dynamics of interaction between the nodes involve local majority-based changes of states, e.g., as in the voter model~\cite{holley1975ergodic,donnelly1983finite} or in the majority dynamics~\cite{krapivsky2003dynamics}.
A substantial line of research has been devoted to study the use of such simple dynamics as lightweight distributed algorithms to solve complex tasks, mirroring the behavior of complex systems from which they take inspiration. Here we are interested in discussing some of the contributions among the large body of work on dynamics for consensus. The reader is deferred to~\cite{becchetti2020overview} for a more detailed survey on the topic.
All dynamics taken into consideration share a common communication model, where nodes can pull information from some fixed number of neighbors before updating their state.

The \voter{} is the first---and arguably the simplest---dynamics considered for consensus.
Hassin and Peleg~\cite{hassin2001voter} proved that the dynamics reaches a consensus on state $\sigma$ with probability proportional to the volume of nodes initially in state $\sigma$ in $\bigO(n^3 \log n)$ rounds, regardless of the graph structure, where $n$ is the number of nodes in the network.
Recently the upper bound has been improved to $\bigO(n^3)$~\cite{DBLP:conf/soda/KanadeMS19}, which is shown to be tight.
However, the dynamics is slow in reaching consensus, e.g., it needs $\Omega(n)$ rounds in the complete graph despite the extremely good connectivity properties of the topology.

Therefore, simple generalizations of the \voter{} have been considered in order to achieve a faster distributed algorithm for consensus. 
One of the directions has been that of considering more than a single neighbor in the sample. 
For example, in the $3$-\maj{} dynamics~\cite{DBLP:journals/dc/BecchettiCNPST17}, the time needed to reach a consensus on the complete graph lowers from $\Omega(n)$ to $\bigO(\log n)$~\cite{DBLP:journals/siamdm/CooperEOR13,DBLP:conf/podc/GhaffariL18}.

The \choices{} dynamics is a variation of the \voter{} in which nodes sample two random neighbors and update their states to the majority among two, breaking ties toward their own states.
The dynamics has been studied with opinions on $d$-regular and expander graphs~\cite{DBLP:conf/icalp/CooperER14}, proving that, given a sufficient initial unbalance between the two opinions, a consensus on the initial majority is reached within a polylogarithmic number of rounds, with high probability. 
Such results have been later improved in~\cite{DBLP:conf/wdag/CooperERRS15}, relaxing the assumptions on graph's expansion, and generalized to the case of multiple opinions~\cite{DBLP:conf/wdag/CooperRRS17,DBLP:conf/podc/ElsasserFKMT17}. 
More recently, the \choices{} dynamics has been analyzed on networks with a core-periphery structure~\cite{DBLP:conf/atal/CrucianiNNS18}, where, depending on the initialization, it exhibits a phase transition phenomenon.

The $3$-\maj{} dynamics is faster than the \voter{} in reaching consensus on well connected graphs, e.g., on the complete graph the time lowers from $\Omega(n)$ to $\bigO(\log n)$~\cite{DBLP:journals/siamdm/CooperEOR13,DBLP:conf/podc/GhaffariL18}.
On the complete graph and with $h$ possible opinions, the process converges to a plurality consensus in $\bigO(\min\{h, \sqrt[3]{n/\log n}\} \cdot \log n)$ with high probability, if the initial unbalance between the plurality color and the second one is large enough~\cite{DBLP:journals/dc/BecchettiCNPST17}.
In~\cite{DBLP:conf/podc/BerenbrinkCEKMN17} unconditional lower and upper bounds for \choices{} and 3-\maj{} on the complete graph are provided, whenever the number of initial colors is large.
The scenario in which an adversary can modify some of the $h$ opinions, again for \choices{} and/or 3-\maj{}, is considered in~\cite{DBLP:conf/soda/BecchettiCNPT16,DBLP:conf/spaa/DoerrGMSS11,DBLP:conf/podc/GhaffariL18}, with the best result proving convergence to a valid consensus in $\bigO(h \log n)$ rounds, with high probability, even if the adversary can control $o(\sqrt{n})$ nodes.
The $3$-\maj{} dynamics has been recently analyzed also on non complete topologies~\cite{kang2019bestof3}. The authors consider a random initialization in which every node is \emph{red} with probability $1/2 + \delta$ and \emph{blue} otherwise and graphs with minimum degree $d = \Omega(n^{1/\log \log n})$. Their result implies, e.g., a consensus on state red in $\bigO(\log \log n)$ rounds, w.h.p., if $\delta = \Omega(1/\log \log n)$.

To the best of our knowledge, $k$-\maj{} has not been extensively studied for generic $k$.
Among the few papers that consider it, in~\cite{DBLP:journals/dam/AbdullahD15} such a dynamics (sampling without replacement) is analyzed for $k \geq 5$ on the configuration model. 
Differently from our contribution, the paper analyzes the process on sparse graphs with low-degree sequences, namely with maximum degree that is sublinear in $n$ and the average degree that is $o(\sqrt{\log n})$.
Similarly to our paper, the the real evolution of the dynamics on the graph is ``approximated'' by the dynamics on some regular trees. In particular, the analysis relies on a coupling between the real process and a process on regular trees, exploits the sparsity of the graph to ensure many local tree-like structures, and makes use of other structural properties to handle those vertices which belong to short cycles. 
The paper proves that the process converges to the initial majority within $\bigO(\log \log n)$ steps, with high probability.
In~\cite{DBLP:conf/waw/AbdullahBF15}, instead, the dynamics is analyzed on preferential attachment graphs with similar convergence upper bounds to~\cite{DBLP:journals/dam/AbdullahD15}.
In~\cite{shimizu2020quasi} a new model is proposed, which contains majority rules as special cases. In particular, for $k$-\maj{} with odd $k$ and in a binary state setting, the convergence time on expander graphs is proved to be $\bigO(\log n / \log k)$ rounds for $k = o(n / \log n)$.

The \determ{} \maj{} differs from the other dynamics since there is no randomness in the interaction, which is a fundamental feature of the dynamics previously discussed. 
This deterministic protocol has been extensively studied in the literature; 
we mention, for example, its analysis on 
expander graphs~\cite{DBLP:journals/aamas/MosselNT14, DBLP:journals/dam/Zehmakan20}, 
random regular graphs~\cite{DBLP:conf/latin/GartnerZ18}, and
Erd\H{o}s-Rényi random graphs~\cite{benjamini2016convergence, DBLP:journals/dam/Zehmakan20}.

\subsection{Biased Communication Models in Opinion Dynamics}
A different perspective coming from other disciplines, such as economics and sociology, is that of considering interaction models between nodes of a network as models of opinion diffusion.
The main models, e.g., DeGroot~\cite{degroot1974reaching} and Friedkin--Johnsen~\cite{friedkin1990social}, are based on averaging dynamics, i.e., nodes move toward the average opinion seen in their neighborhoods.
Nevertheless, also the dynamics previously discussed can be framed in the modeling of opinion dynamics (see, e.g., \cite{amir2019majority,DBLP:conf/ijcai/AulettaFG18,mei2019median,mossel2017survey,DBLP:books/sp/17/SirbuLST17}).

Some opinion dynamics have been considered with biased communication models, specially in asynchronous case.
The binary \determ{} \maj{} dynamics has been considered in a setting where nodes have a fixed private opinion and, when active, announce a public opinion as the majority opinion in their neighborhood, but ties are broken toward their private belief.
Such a process has been proved to converge to the initial private majority whenever the graphs are sufficiently sparse and expansive~\cite{DBLP:conf/approx/FeldmanILW14} or preferential attachment trees~\cite{DBLP:conf/icalp/immorlica19}.
The binary \determ{} \maj{} dynamics, as well as the binary \voter{}, have also been analyzed in asynchronous models presenting different forms of bias~\cite{mukhopadhyay2016binary,DBLP:conf/ijcai/Anagnostopoulos20}.
In~\cite{mukhopadhyay2016binary}, if the network prefers, say, opinion $a$ instead of $b$, every node holding opinion $b$ updates more frequently than the others; this particular feature is modeled by allowing nodes in state $b$ to revise their opinion at all points of a Poisson process with rate $q_b > q_a$.

In~\cite{DBLP:conf/ijcai/Anagnostopoulos20,lesfari2022biased} the bias is defined toward a superior opinion: nodes have a fixed probability $p$ of updating their state to such an opinion, independently of the dynamics. The analysis in~\cite{DBLP:conf/ijcai/Anagnostopoulos20} is relative to the \emph{asynchronous} \voter{} and \determ{} \maj{} dynamics. The paper proves that network density negatively affects the speed of convergence to the superior opinion when the dynamics is \determ{} \maj{}, while it does not play any role when the dynamics is the \voter{}.
New results for sparse graphs in the same model of bias have been provided in~\cite{lesfari2022biased}, proving a polynomial-time convergence to the superior opinion in cubic graphs and an exponential-time convergence to the superior opinion in sparse (constant degree) random regular graphs, for suitable values of the bias.

Another variant of bias in the \voter{} is analyzed on static and dynamic graphs (with the constrain that the degrees of the nodes do not change over time) in~\cite{berenbrink_et_al:LIPIcs:2016:6290}. 
There are $\ell$ opinions associated with popularities $1=\alpha_1>\alpha_2\ge\ldots\ge\alpha_\ell\ge 0$, so that the preferred opinion $1$ has strictly larger popularity than the others (i.e., $\alpha_1=1$ w.l.o.g.).
In each round each node samples a random neighbor and adopts its opinion $i$ with probability $\alpha_i$.
On static graphs, the paper proves that, if the preferred opinion is initially supported by $\Omega(\log n)$ nodes and $\alpha_2 \le 1-\epsilon$ for some positive constant $\epsilon$, then the process converges to a consensus on the preferred opinion with high probability within $O(\log n / \phi)$ rounds, with $\phi$ being the conductance of the graph.

A biased version of the \voter{} has also been studied in a synchronous model in~\cite{DBLP:conf/icalp/BerenbrinkGKM16}: the nodes, after selecting a random neighbor, have a probability of copying its state that depends on the state itself.

In~\cite{DAmoreCN20} the binary \undecided{} dynamics, a variation of the \voter{} using one extra state (the undecided one) and introduced in~\cite{angluin2008simple}, is analyzed in the complete graph in a communication model presenting a uniform noise, i.e., every state can be confused with the opposite with probability $p$. The authors show a phase transition: if $p<\frac{1}{6}$ the process rapidly reaches a metastable regime of almost-consensus which lasts for polynomial time; otherwise, if $p>\frac{1}{6}$, the initial majority is lost within a logarithmic number of rounds.
In the same communication model with uniform noise and in the same complete topology, also the 3-\maj{} dynamics has been analyzed, exhibiting the same phase transition on a critical value $p=\frac{1}{3}$~\cite{damore:hal-03467403}.
\section{Preliminaries}\label{sec:model}
\subsection{Notation and Computational Model}\label{sec:notation}

Consider a sequence of graphs $(G_n)_{n\in\mathbb{N}}$, where $G_n=(V_n,E_n)$ and $V_n = \{1,\dots,n\}$. 
Note that the graphs we consider are not necessarily connected and can also be directed.
We are interested in the asymptotic case in which $n\to\infty$ and therefore, to simplify the notation, we will usually drop the dependence on $n$. For the same reason, we will often use the Bachmann--Landau notation (i.e., $\omega, \Omega, \Theta, \bigO, o$) to describe the limiting behavior of functions depending on $n$.

For each node $u \in V$, let $\neigh{u} \dfn \{v \in V : (u,v) \in E\}$ be the \emph{neighborhood} of $u$. 
In the following, we focus only on \emph{sufficiently dense} graphs, i.e., graphs where every node $u\in V$ has \emph{degree} $\deg{u} \dfn |\neigh{u}| = \omega(\log n)$.%
\footnote{We refer to such graphs as ``sufficiently dense'' in order to highlight the lower bound on the degree of the nodes that is necessary for our proof techniques.}
In the case of directed graphs, the restriction is on the out-degree of nodes, namely the number of outgoing directed edges.

In this paper, a dynamics on a given graph $G$ is a process that evolves in discrete, synchronous rounds,%
\footnote{Equivalently, nodes have access to a shared clock.}
where, in every round $t \in \mathbb{N}_0$, every node $u \in V$ has a binary state $\state{u}{t} \in \{\red,\blue\}$ that can change over time according to a function of the states of its neighbors;
we denote the \emph{configuration} of the system at round $t$, i.e., the vector of states of the nodes of $G$, as $\conf{t} \in \{\red,\blue\}^n$;
we define 
\begin{equation}\label{rtbt}
\R{t} \dfn \{u \in V : \state{u}{t} = \red\}\qquad\text{and}\qquad\B{t} \dfn \{u \in V : \state{u}{t} = \blue\}\,.
\end{equation}

We consider two different dynamics on the graph $G$, namely the \majority{k}{p}{\blue} dynamics (\cref{def:majority}) and the \modmajority{k}{p}{\blue} dynamics (\cref{def:modmajority}). 
Both processes are obtained from the $k$-\maj{} dynamics by introducing a \emph{bias} $p$ toward one of the two states, which we assume to be the state $\blue$.
Recall that in the $k$-\maj{} dynamics, in each round, every node samples $k$ neighbors uniformly at random and with replacement; then it updates its state to the state held by the majority of the neighbors in the sample in the previous round; ties are broken uniformly at random.
Differently, in the \majority{k}{p}{\blue} dynamics, whenever node $u$ samples a neighbor $v$, the state of $v$ can be seen by $u$ as \blue with probability $p$ regardless of its actual state (as if an external agents corrupts the communication that is modeled through Z-channels~\cite{mackay2003information}). 
Instead, in the \modmajority{k}{p}{\blue}, whenever node $u$ updates its state, it has a probability $p$ of choosing state $\blue$ regardless of the actual majority in the sample, i.e., the nodes spontaneously change state to \blue with probability $p$ (as if corrupted by an external agent).

With a slight abuse of notation we let $\mathbf{P}=\mathbf{P}^{(n)}$ denote both the law of the \majority{k}{p}{\blue} and that of the \modmajority{k}{p}{\blue} dynamics on the graph $G=G_n$, taking care of specifying the dynamics to which the law refers to in the rest of the paper. 
We use the notation $\mathbf{1}_{\mathcal{A}}$ for the indicator variable of the event $\mathcal{A}$, i.e., $\mathbf{1}_{\mathcal{A}} = 1$ if $\mathcal{A}$ holds and $\mathbf{1}_{\mathcal{A}} = 0$ otherwise.
Formally, \majority{k}{p}{\blue} can be described as follows.
Let $p \in [0,1]$ be the parameter that models the \emph{bias} toward state $\blue$; 
let $k \in \mathbb{N}$ be the size of the sampling.
For each round $t$, let $\sample{u}{t}$ be the multiset%
\footnote{Recall that the sampling is with replacement.}
of neighbors sampled by node $u$ in round $t$.
For each sampled node $v \in \sample{u}{t}$, we call $\bstate{u}{v}{t}$ the state in which node $u$ sees $v$ after the effect of the bias, i.e., $\bstate{u}{v}{t} = \blue$ if $\state{v}{t}(u)= \blue$, and $\bstate{u}{v}{t} = \blue$ with probability $p$ otherwise; formally
\[
    \Prob{\bstate{u}{v}{t} = \blue  \condconf} = \mathbf{1}_{\state{v}{t}=\blue} + p\cdot\mathbf{1}_{\state{v}{t}\neq\blue}\,.
\]
We define 
\begin{equation}\label{barsets}
    \bR{u}{t} \dfn \{v \in \sample{u}{t} : \bstate{u}{v}{t} = \red\}
    \qquad\text{and}\qquad
    \bB{u}{t} \dfn \{v \in \sample{u}{t} : \bstate{u}{v}{t} = \blue\}
\end{equation}
as the sets of sampled nodes that $u$ sees respectively in state \red{} and \blue{} after the effect of the bias $p$.

\begin{definition}[\majority{k}{p}{\blue} dynamics]\label{def:majority}
Let $p \in [0,1]$ and $k\in \mathbb{N}$. Starting from an initial configuration $\conf{0}$, at each round $t$ every node $u \in V$ updates its state as
\[
    \state{u}{t} = \left\{\begin{array}{ll}
        \red
            & \text{if } |\bR{u}{t-1}| > |\bB{u}{t-1}|\,,
        \\
        \red \text{ or } \blue \text{ with probability } 1/2
            & \text{if } |\bR{u}{t-1}| = |\bB{u}{t-1}|\,,
        \\
        \blue
            & \text{if } |\bR{u}{t-1}| < |\bB{u}{t-1}|\,.
    \end{array}
    \right.
\]
\end{definition}

Similarly, the \modmajority{k}{p}{\blue} can be formally defined as follows.
We define 
\[
    R_u^{(t)}\dfn \{v \in \sample{u}{t} : \state{v}{t} = \red\}
    \qquad\text{and}\qquad
    B_u^{(t)} \dfn \{v \in \sample{u}{t} : \state{v}{t} = \blue\}
\]
respectively as the set of neighbors in state \red{} and in state \blue{} that node $u$ samples at round $t$; 
for every $u \in V$ and every $t \in \mathbb{N}$, denote by $M_u^{(t)}$ the Bernoulli random variable of parameter $p$ which models the bias toward node $u$ in round $t$, i.e., $M_u^{(t)}=1$ with probability $p$ and $M_u^{(t)}=0$ otherwise.
\begin{definition}[\modmajority{k}{p}{\blue} dynamics]\label{def:modmajority}
Let $p \in [0,1]$ and $k\in \mathbb{N}$. Starting from an initial configuration $\conf{0}$, at each round $t$ every node $u \in V$ updates its state as
\[
    \state{u}{t} = \left\{\begin{array}{ll}
        \red & \text{if } M_u^{(t)}=0 \text{ and } |R_u^{(t-1)}|>|B_u^{(t-1)}|\,,
        \\
        \red \text{ or } \blue \text{ with probability $1/2$} 
            & \text{if } M_u^{(t)}=0 \text{ and } |R_u^{(t-1)}|=|B_u^{(t-1)}|\,,
        \\
        \blue & \text{if } M_u^{(t)}=1 \,\,\text{ or }\,\, |R_u^{(t-1)}|<|B_u^{(t-1)}|\,.
    \end{array}
    \right.
\]
\end{definition}

Note that the \majority{k}{p}{\blue} dynamics and the \modmajority{k}{p}{\blue} are Markov chains, since the configuration $\conf{t}$ in a round $t>0$ depends only on the configuration at the previous round, i.e., $\conf{t-1}$.
Moreover, when $p>0$, they both have a single absorbing state in which $\state{u}{t} = \blue$ for every $u \in V$.
Indeed, being the configuration space finite for every choice of $n$, it is easy to prove that the probability of jumping to such a configuration is positive when $p$ is positive, e.g., giving a simple lower bound of $p^{kn}$ to the event that all nodes see all their neighbors in state \blue{} in any given round for the  \majority{k}{p}{\blue} dynamics or $p^{n}$ for the event that all neighbors directly update to state \blue for the \modmajority{k}{p}{\blue} dynamics.
Moreover, such a configuration is the unique absorbing state. Indeed, once such a configuration in reached, none of the nodes can change its state to \red{} at any subsequent round.

We consider initial configurations in which each node is in state \red, independently of the others, with probability $q>\frac{1}{2}$. 
This choice of $q$ guarantees, by a concentration of probability argument, an initial majority of nodes in state \red with high probability.
The two dynamics that we analyze allow to study the \robustness of the $k$-\maj{} when either the communication between nodes or the nodes themselves can be corrupted.
In particular, both dynamics model an external, adversarial agent that tries to subvert the initial majority on state \red{} toward state \blue{}.
For this reason, as previously announced, we give the following definition of \emph{time of disruption} $\tau$.
\begin{definition}[Time of disruption]\label{def:time_disruption}
The \emph{time of disruption} of the \red majority is the stopping time
\[
    \tau:=\inf\left\{
        t\ge 0 \cond \frac{|\B{t}|}{n}>\frac{1}{2}
    \right\}\,.
\]
\end{definition}
In other words, the time of disruption $\tau$ is the first time at which the initial \red majority is subverted by state \blue.
Note that we fix the threshold for disruption at $\frac{1}{2}$ for the sake of simplicity, but our results hold for any constant arbitrarily close to $1$.

In the following sections (and wherever previously used), we say that an event $\mathcal{E}_n$ holds \emph{with high probability} (\emph{w.h.p.}, in short) if $\Prob{\mathcal{E}_n} = 1 - o(1)$. In this sense, our results only hold for large $n$.
We also use the notation $\bin{n}{p}$ to indicate a random variable sampled from the Binomial distribution of parameters $n$ (number of trials) and $p$ (probability of success).

\subsection{On Even Values of the Sample Size \texorpdfstring{$k$}{k}}
In the rest of the paper we assume samples of odd size $k$, thus avoiding potential ties. 
In this section we show that this assumption is not restrictive.
In fact, in \cref{lem:even-odd} we prove the equivalence between the Markov chains \majority{2h+1}{p}{\blue} and \majority{2h+2}{p}{\blue} (respectively, for the Markov chains  \modmajority{2h+1}{p}{\blue} and \modmajority{2h+2}{p}{\blue}).
The result in \cref{lem:even-odd} is essentially a special case of~\cite[Appendix B]{DBLP:journals/dc/FraigniaudN19}, but it has been obtained independently and with a simplified proof.

\begin{proposition}
\label{lem:even-odd}
Let $G=(V,E)$ be a graph with binary state configuration $\conf{t}$ in round $t$. Fix some $k \in \mathbb{N}$ and $p \in [0,1]$ and consider the \majority{k}{p}{\blue} (\modmajority{k}{p}{\blue}, respectively) dynamics.
For each node $u \in V$ define $\Emaj_k$ the event that node $u$ is in state $\red$ at round $t+1$, i.e.,
\[
    \Emaj_k\coloneqq \{\state{u}{t+1}=\red\}\,.
\]
Then, for every $h \in \mathbb{N}_0$ it holds
\[
    \Prob{\Emaj_{2h+1}  \condconf} = \Prob{\Emaj_{2h+2} \condconf}\,.
\]
\end{proposition}
\begin{proof}
   Fix any $t\ge 0$. We first consider a single round of the \majority{k}{p}{\blue} dynamics conditioned on the configuration at step $t$, i.e., $\conf{t}$. In order to ease the reading we omit the explicit conditioning on $\conf{t}$ in our notation. 
    Consider the event $\Emaj_k$ as
    ``Node $u$ is in state \red at round $t+1$''.
	 Let us focus on samples of odd size, i.e., $k=2h+1$, where there cannot be ties.
	Let us define $z\dfn (1-p) \frac{|\R{t}_u|}{\deg{u}}$, i.e., the probability that $u$ sees a given node in the sample in state \red, and the random variable $X_R=\bin{2h+1}{z}$
	modeling the number of nodes seen in state \red out of a sample of size $2h+1$.
	Since the samples are independent, one can study the event $\Emaj_{2h+2}$ by coupling it with the event $\Emaj_{2h+1}$ and performing an additional sample (with $z$ as success probability) after the first $2h+1$. For this purpose, we also define the random variable $Y_R=\bin{2h+2}{z}$.
	Therefore:
	\begin{enumerate}
	    \item \( \Prob{\Emaj_{2h+1}} = 
			\Prob{X_R \geq h+1} =
			\Prob{X_R = h+1} + \Prob{X_R > h+1}\,, \)
		\item \( \Prob{\Emaj_{2h+2}} = 
		    \Prob{Y_R > h+1} + \frac{1}{2}\Prob{Y_R = h+1} 
		    \\
			= \Prob{X_R > h+1} +
			\Prob{X_R = h+1} z +
			\frac{1}{2} \left[
				\Prob{X_R = h+1} \left(1-z\right) +
				\Prob{X_R = h} z
			\right]\,. \)
	\end{enumerate}
	Moreover, note that
    \(
	    \Prob{X_R = h} \cdot z
	    = \binom{2h+1}{h} \left(1-z\right)^{h+1} z^{h+1}
		= \Prob{X_R = h+1} \cdot \left(1-z\right).
	\)
	Therefore we conclude the proof by plugging the latter equivalence into the previous formulation of $\Prob{\Emaj_{2h+2}}$, getting that
	\( \Prob{\Emaj_{2h+1}} = \Prob{\Emaj_{2h+2}} \).
	
	In the case of the \modmajority{k}{p}{\blue} dynamics we can proceed in a similar way. Indeed, defining $z\dfn\frac{|\R{t}_u|}{\deg{u}}$, we get
	\begin{enumerate}
	    \item \( \Prob{\Emaj_{2h+1}} = 
			(1-p)\Prob{X_R \geq h+1} =
			(1-p)\left[\Prob{X_R = h+1} + \Prob{X_R > h+1}\right]\,, \)
		\item \( \Prob{\Emaj_{2h+2}} = 
		    (1-p)\left[\Prob{Y_R > h+1} + \frac{1}{2}\Prob{Y_R = h+1}\right]
		    \\
			= (1-p)\left[\Prob{X_R > h+1} +
			\Prob{X_R = h+1} z +
			\frac{1}{2} \left[
				\Prob{X_R = h+1} \left(1-z\right) +
				\Prob{X_R = h} z
			\right]\right]\,. \)
	\end{enumerate}
    Hence we can conclude as in the previous case.
\end{proof}

\begin{corollary}
Consider a graph with any initial configuration $\conf{0}$ and fix any $p \in [0,1]$.
 For every $h \in \mathbb{N}_0$, the \majority{2h+1}{p}{\blue} (\modmajority{2h+1}{p}{\blue}, respectively) dynamics and the \majority{2h+2}{p}{\blue} (\modmajority{2h+2}{p}{\blue}, respectively) dynamics follow the same law.
\end{corollary}
\begin{proof}
By \cref{lem:even-odd} and the independence in the updates, at each fixed round $t$, conditioned to the configuration at round $t-1$, the dynamics with $k=2h+1$ and $k=2h+2$ have the same law, for any integer $h\geq 1$. As a consequence, using the chain rule and the Markov property, this property can be extended for every finite set of rounds. This concludes the proof since the absorbing time is almost surely finite.
\end{proof}

\section{The Tree-Processes}\label{sec:preli}
Consider a given node $v\in V$ of a graph $G$. The state of $v$ in round $t$ is a random variable which is measurable with respect to the state in round $t-1$ of the $k$ neighbors that $v$ samples in round $t$, say $u^{1}_1,\dots, u^{1}_k$. 
Similarly, the states in round $t-1$ of $u^{1}_1,\dots,u^{1}_k$ depend only on the states in round $t-2$ of the $k^2$ neighbors that they sample in round $t-1$, say $u^{2}_1,\dots, u^{2}_{k^2}$. 
Iterating this argument, we end up concluding that the state of $v$ in round $t$ depends only on the states of $\{u_{1}^s,\dots, u_{k^s}^s\}_{s=1,\dots,t}$. 
Let us assume for the moment that $\{u_{1}^s,\dots, u_{k^s}^s \}_{s=1,\dots,t}$ is a collection of distinct vertices. 
Then, in order to sample the color of $v$ in round $t$ it is enough to consider the first $t$ layers of an infinite $k$-regular directed tree having $v$ as a root and $\{u^s_{1},\dots, u^s_{k^s} \}$ in the $s$-th layer, for $s\le t$. 
On such a directed tree we consider the analogue of the \majority{k}{p}{\blue} (respectively, \modmajority{k}{p}{\sigma}), which is defined as follows.

Let $\mathcal{T}=(V,E)$ be our $k$-regular infinite tree with edges oriented toward the children.
We call $\vroot \in V$ the root of $\mathcal{T}$.
For each node $u\in V$ we define the binary state $x_u^{(t)}\in\{\red,\blue\}$ that can change over time as described in \cref{def:majority} (for the \majority{k}{p}{\blue}) and in \cref{def:modmajority} (for the \modmajority{k}{p}{\blue}), but with the only difference that $\sample{u}{t}$ will be the set of $k$ children of node $u$. 
It is clear by the argument at the beginning of this section that the \emph{tree-process} can be coupled with the process on the graph in a way that if the vertices $\{u^s_{1},\dots, u^s_{k^s} \}_{s=1,\dots,t}$ are all distinct then the color of the root in round $t$ coincides with the color of the vertex $v$ in the graph at the same time.

In what follows we analyze such a \emph{tree-process} and completely characterize its behavior. In the forthcoming \cref{sec:phase_transition} we show that the quantities appearing in the analysis of this process can be recovered by analyzing the expected density of vertices in state \red at a given time $t$.

In order to make the exposition clearer, in what follows we separate the analysis of the \emph{tree-process} for the two dynamics in separate subsections.

\subsection{\texorpdfstring{\majority{k}{p}{\blue}}{(k,p,B)-Edge Majority} Dynamics}
As discussed in \cref{sec:model}, let us define the binary random variable $\bstate{u}{v}{t}$ for each round $t\ge0$, each vertex $v$, and each child $u$ of $v$. In each round $t$ we construct the sets $\R{t}$ and $\B{t}$ as described in \cref{sec:model}.
We study the evolution of the probability of the event ``$\vroot \in \R{t}$'', i.e., the root of $\mathcal{T}$ is in state \red at round $t$.
The result in the forthcoming \cref{th:tree} is based on the analysis of the function $\Fpk{p}{k}$, described by the following definition, that represents the evolution of the event under analysis. 
To ease the intuition, if at any time $t$ each vertex has probability $x$ of being in state \red, then the probability that a given vertex is in state \red at the next step is given by $\Fpk{p}{k}(x)$.
\begin{definition}[Function $\Fpk{p}{k}$]\label{def:fpk}
Let $h \in \mathbb{N}$ and let $k \dfn 2h+1$. Let $p \in [0,1]$.
We define the function $\Fpk{p}{k} : [0,1] \rightarrow [0,1]$ as
\[
    \Fpk{p}{k}(x) \dfn \Prob{\bin{k}{(1-p)x}\geq \frac{k+1}{2}}\,.
\]
\end{definition}

In particular, we will use the following facts about $\Fpk{p}{k}$, which are proved in \cref{sec:apx fpk} and depicted in \cref{fig:qualplot} (left plot).
\begin{restatable}{lemma}{fixedpoints}
\label{lem:fixed_points_F-prel}
	For every finite odd $k\geq 3$, there exists $\pk \in \left[\frac{1}{9},\frac{1}{2}\right)$ such that:
	\begin{itemize}
		\item if $p < \pk$, there exist $\varphi^-_{p,k},\varphi^+_{p,k} \in \big(\frac{1}{2(1-p)}, 1\big]$ with $\varphi^-_{p,k}<\varphi^+_{p,k}$ such that $\Fpk{p}{k}(x) = x$ has solutions $0, \varphi^-_{p,k}$, and $\varphi^+_{p,k}$. 
		Moreover, $\Fpk{p}{k}(x)-x< 0$ for $x\in(0,\varphi^-_{p,k})\cup (\varphi^+_{p,k},1]$ while $\Fpk{p}{k}(x)-x> 0$ for $x\in(\varphi^-_{p,k},\varphi^+_{p,k})$;
		\item if $p = \pk$, there exists $\varphi_{p,k} \in \big(\frac{1}{2(1-p)}, 1 \big]$ such that $\Fpk{p}{k}(x) = x$ has solutions $0$ and $\varphi_{p,k}$. Moreover $\Fpk{p}{k}(x) < x$ if $x\not\in\{0, \varphi_{p,k}\}$;
		\item if $p > \pk$, then $\Fpk{p}{k}(x) = x$ has $0$ as unique solution. Moreover, $\Fpk{p}{k}(x) < x$ if $x\neq 0$.
	\end{itemize}
	Furthermore, the sequence $\{\pk\}_{k\in 2\mathbb{N}+1}$ is increasing.
\end{restatable}
Note that it is not possible to give a closed formula of $\pk$ for generic $k$ because it is the root of a polynomial of degree $k$; numerical approximations can be computed for any given $k$. However, as proven in \cref{lem:fixed_points_F-prel}, $\pk$ monotonically increases with $k$, starting from $\frac{1}{9}$ (for $k=3$) and up to $\frac{1}{2}$ (its limit value as $k\rightarrow\infty$, as proven in \cref{claim:05}).

\begin{figure}[ht]
\begin{center}\hfill%
    \begin{minipage}[t]{0.5\textwidth}
      \includegraphics[width=\linewidth]{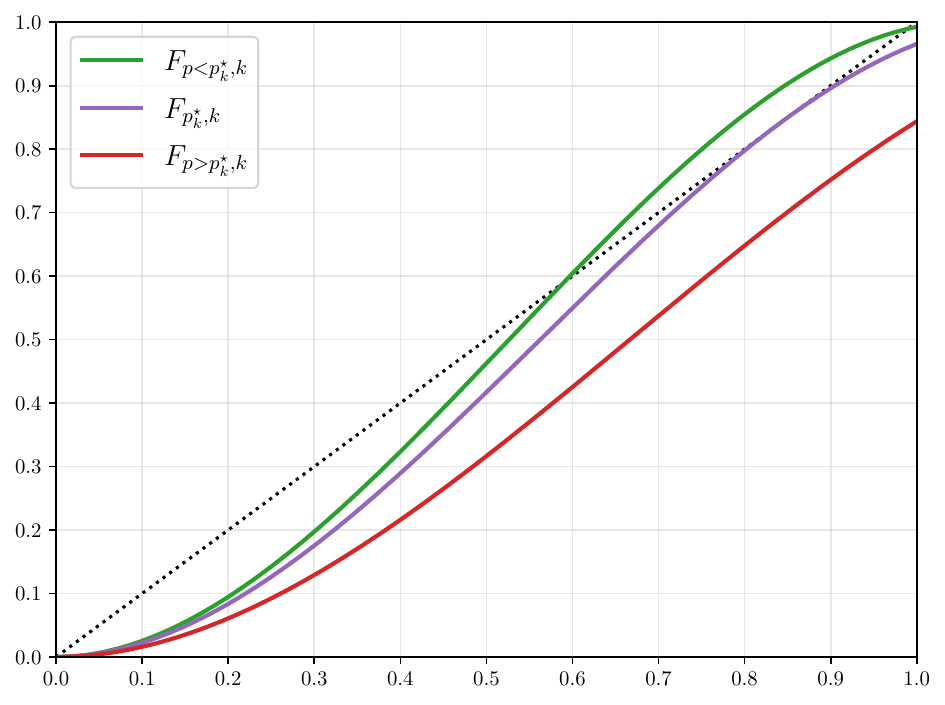}
    \end{minipage}\hfill%
    \begin{minipage}[t]{0.5\textwidth}
      \includegraphics[width=\linewidth]{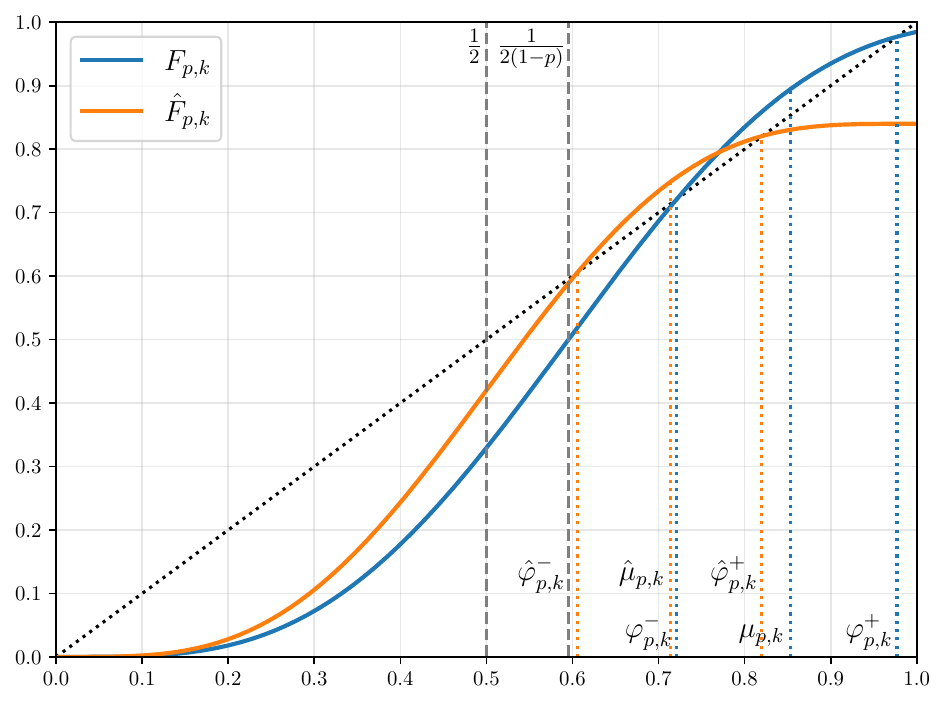}
    \end{minipage}\hfill%
\end{center}
\caption{\footnotesize
Qualitative plots of the functions $\Fpk{p}{k}(x)$ and $\widehat{F}_{p,k}(x)$.
On the left we focus on $\Fpk{p}{k}(x)$, for $k=3$ and with three different values of $p$, namely $p=\frac{1}{20}<\pk$, $p=\frac{1}{9}=\pk$, and $p=\frac{1}{4}>\pk$.
On the right we compare $\Fpk{p}{k}(x)$ and $\widehat{F}_{p,k}(x)$, for $k=7$ and $p=0.16<\pk$, and show the points of interest of the two functions, namely the two inflection points in $\frac{1}{2(1-p)}$ and $\frac{1}{2}$, the fixed points $\varphi^-_{p,k}$ and $\varphi^+_{p,k}$, and the unique point $\mu_{p,k}>\varphi^-_{p,k}$ for which $\Fpk{p}{k}^{\,'}(\mu_{p,k})=1$ (respectively, $\hat{\varphi}^-_{p,k}$ and $\hat{\varphi}^+_{p,k}$, and $\hat{\mu}_{p,k}>\hat{\varphi}^-_{p,k}$
for $\widehat{F}_{p,k}(x)$).}
\label{fig:qualplot}
\end{figure}

We now prove the following theorem, which will be a guideline for the whole paper.
\begin{theorem}\label{th:tree}
Consider the \majority{k}{p}{\blue} dynamics on $\mathcal{T}$, where at round $t=0$ each vertex of the tree is in state $\red$ with probability $q\in[0,1]$ or in state \blue with probability $1-q$, independently of the others. Define $q_t=\Prob{\vroot \in R^{(t)}}$ for $t\in\mathbb{N}_0$. Then, fixed the value of $q_0=q\in[0,1]$, the sequence $(q_t)_{t\in\mathbb{N}_0}$ is monotone and  can be rewritten as
\begin{equation}\label{eq:def-q}
    q_{t+1}=\begin{cases}
\Fpk{p}{k}(q)      &\text{if } t=0\,,
\\
\Fpk{p}{k}(q_t)    &\text{if } t\geq 1\,.
\end{cases}
\end{equation}
Moreover we have
\begin{equation}\label{eq:th-tree}
\lim_{t\to\infty}q_t=
\begin{cases}
\varphi^+_{p,k}&\text{if }p<\pk\text{ and }q>\varphi^-_{p,k}\,,
\\
0&\text{if }p<\pk\text{ and }q<\varphi^-_{p,k}\,,
\\
0&\text{if }p>\pk\,.
\end{cases}
\end{equation}
\end{theorem}

\begin{proof}
Define $q_t:=\Prob{\vroot\in \R{t}}$
and observe that, since we are working on an infinite tree, we have 
$\Prob{\vroot \in \R{t}} = \Prob{v\in \R{t}}$ for all $v\in V$ and for all $t\in\mathbb{N}_0$.
Note also that, for any siblings $v,w\in V$, the events ``$v\in \R{t}$'' and ``$w\in \R{t}$'' are independent at any round $t\in\mathbb{N}_0$. 
Hence the random variables in the family
$\{\mathbf{1}_{v\in \R{t}} : v\in V\}$ are i.i.d.\ and $\Prob{v\in \R{t}}=q_t$.

Let us now compute $q_{t+1} = \Prob{\vroot \in \R{t+1}}$. 
We have:
\[
    q_{t+1}=\mathbf{P}\left(
        \sum_{w\in \sample{\vroot}{t}}
        \mathbf{1}_{w\in \R{t}}
        \mathbf{1}_{\bstate{\vroot}{w}{t} = \state{w}{t}}
        \geq \frac{k+1}{2} 
    \right)\,.
\]
Note that the family of random variables 
$\left\{ \mathbf{1}_{
    \bstate{\vroot}{w}{t} = \state{w}{t}
} \cond w\in \sample{\vroot}{t}\right\}$ 
is independent of the family of random variables
$\left\{\mathbf{1}_{w\in \R{t}} \cond w\in \sample{\vroot}{t}\right\}$.
Therefore $\left\{
\mathbf{1}_{w\in \R{t}}
\mathbf{1}_{
    \bstate{\vroot}{w}{t} = \state{w}{t}
} \cond w\in \sample{\vroot}{t}
\right\}$ 
is a family of i.i.d.\ Bernoulli random variables of parameter $(1-p)q_t$
and hence
\[
    \sum_{w\in \sample{\vroot}{t}}
    \mathbf{1}_{w\in \R{t}}
    \mathbf{1}_{\bstate{\vroot}{w}{t}= \state{w}{t}}
    \overset{d}{=} \bin{|\sample{\vroot}{t}|}{(1-p)q_t}
    \overset{d}{=} \bin{k}{(1-p)q_t}\,.
\]
Thus, for every $t \geq 0$, we can write 
\begin{equation}\label{prelim:fpk}
q_{t+1}
=\Prob{\bin{k}{(1-p)q_t} \geq \frac{k+1}{2}}\,.
\end{equation}
By \cref{def:fpk}, the sequence described in \cref{prelim:fpk} can be rewritten as 
\begin{equation*}
    q_{t+1}=\begin{cases}
\Fpk{p}{k}(q)      &\text{if } t=0\,,
\\
\Fpk{p}{k}(q_t)    &\text{if } t\geq 1\,.
\end{cases}
\end{equation*}
Hence, the limit behavior in \cref{eq:th-tree} and the monotonicity of the sequence $(q_t)_t$ follow from \cref{lem:fixed_points_F-prel}.
\end{proof}

\subsection{\texorpdfstring{\modmajority{k}{p}{\blue}}{(k,p,B)-Node Majority} Dynamics}
As described in \cref{sec:model}, while in \majority{k}{p}{\blue} there is a bias in the communication channels, in \modmajority{k}{p}{\blue} each node can spontaneously change state to $\blue$ with probability $p$.
As for the other model, we study the evolution of the probability of the event ``$\state{v_0}{t}=\red$'', i.e., the root of $\mathcal{T}$ is in state \red at round $t$.
The result in the forthcoming \cref{th:tree2} is based on the analysis of the function $\widehat{F}_{p,k}$, described by the following definition, that represents the evolution of the event under analysis.

\begin{definition}[Function $\widehat{F}_{p,k}$]\label{def:hatfpk}
Let $h \in \mathbb{N}$ and let $k \dfn 2h+1$. Let $p \in [0,1]$.
We define the function $\widehat{F}_{p,k} : [0,1] \rightarrow [0,1]$ as
\[
    \widehat{F}_{p,k}(x)=(1-p)\Prob{\bin{k}{x}\geq \frac{k+1}{2}}\,.
\]
\end{definition}
\noindent
Note that, since
\begin{equation}\label{eq:equivalence}
\widehat{F}_{p,k}(x)=(1-p)\Fpk{p}{k}\left(\frac{x}{1-p}\right)\,,
\end{equation}
the graph of the function $\widehat{F}_{p,k}$ can be obtained from the graph of the function $\Fpk{p}{k}$ through contractions on both the coordinate axes. See \cref{fig:qualplot} (right plot) for a comparison between $\Fpk{p}{k}$ and $\widehat{F}_{p,k}$.

\medskip
 
We will use the following facts about $\widehat{F}_{p,k}$, which are direct consequences of \cref{lem:fixed_points_F-prel} and \cref{eq:equivalence}. 
 
\begin{lemma}
\label{lem:fixed_points_F-prel2}
	For every finite odd $k\geq 3$, there exists $\pk \in \left[\frac{1}{9},\frac{1}{2}\right)$ such that:
	\begin{itemize}
		\item if $p < \pk$, there exist $\hat\varphi^-_{p,k},\hat\varphi^+_{p,k} \in \big(\frac{1}{2}, 1-p\big]$ with $\hat\varphi^-_{p,k}<\hat\varphi^+_{p,k}$ such that $\widehat{F}_{p,k}(x) = x$ has solutions $0, \hat\varphi^-_{p,k}$, and $\hat\varphi^+_{p,k}$. Moreover, $\widehat{F}_{p,k}(x)-x< 0$ for $x\in(0,\hat\varphi^-_{p,k})\cup (\hat\varphi^+_{p,k},1]$ while $\widehat{F}_{p,k}(x)-x> 0$ for $x\in(\hat\varphi^-_{p,k},\hat\varphi^+_{p,k})$;
		\item if $p = \pk$, there exists $\hat\varphi_{p,k} \in \big(\frac{1}{2}, 1-p \big]$ such that $\widehat{F}_{p,k}(x) = x$ has solutions $0$ and $\hat\varphi_{p,k}$. Moreover $\widehat{F}_{p,k}(x) < x$ if $x\not\in\{0, \hat\varphi_{p,k}\}$; 
		\item if $p > \pk$, then $\widehat{F}_{p,k}(x) = x$ has $0$ as unique solution. Moreover $\widehat{F}_{p,k}(x) < x$ if $x\neq 0$.
	\end{itemize}
	Furthermore, the sequence $\{\pk\}_{k\in 2\mathbb{N}+1}$ is increasing. 
\end{lemma}

Note that by \cref{eq:equivalence} the equation $\widehat{F}_{p,k}(x)=x$ is equivalent to the equation $\Fpk{p}{k}(y)=y$, where $y=\frac{x}{1-p}$. Hence the value $\pk$ for the \majority{k}{p}{\blue}, introduced in \cref{lem:fixed_points_F-prel}, coincides with the value $\pk$ for the \modmajority{k}{p}{\blue}, introduced in \cref{lem:fixed_points_F-prel2}. 
\medskip

We now prove the following theorem, which is the analogue of \cref{th:tree} for the \modmajority{k}{p}{\blue}.
\begin{theorem}\label{th:tree2}
Consider the \modmajority{k}{p}{\blue} dynamics on $\mathcal{T}$, where at round $t=0$ each vertex of the tree is in state $\red$ with probability $q\in[0,1]$ or in state \blue with probability $1-q$, independently of the others.
Define $\hat q_t=\Prob{\vroot \in R^{(t)}}$ for $t\in\mathbb{N}_0$. Then, fixed the value of $\hat q_0=q\in[0,1]$, the sequence $(\hat q_t)_{t\in\mathbb{N}_0}$ is monotone and  can be rewritten as
\begin{equation}\label{eq:def-q2}
    \hat q_{t+1}=\begin{cases}
\widehat{F}_{p,k}(q)      &\text{if } t=0\,,
\\
\widehat{F}_{p,k}(\hat q_t)    &\text{if } t\geq 1\,.
\end{cases}
\end{equation}
Moreover we have
\begin{equation}\label{eq:th-tree2}
\lim_{t\to\infty}\hat q_t=
\begin{cases}
\hat \varphi^+_{p,k}&\text{if }p<\pk\text{ and }q>\hat\varphi^-_{p,k}\,,
\\
0&\text{if }p<\pk\text{ and }q<\hat\varphi^-_{p,k}\,,
\\
0&\text{if }p>\pk\,.
\end{cases}
\end{equation}
\end{theorem}

\begin{proof}
Define $\hat q_t:=\Prob{\state{v_0}{t}=\red}$
and observe that, as in the other model, 
$\Prob{\state{v_0}{t}=\red} = \Prob{\state{v}{t}=\red}$ for all $v\in V$ and for all $t\in\mathbb{N}_0$.
 
Let us now compute $\hat q_{t+1}=\Prob{\state{v_0}{t+1}=\red}$. 
Since $|R_{v_0}^{(t)}|+|B_{v_0}^{(t)}|=|S_{v_0}^{(t)}|=k$ and $k$ is odd, we have
\begin{equation}\label{eq:din2-1}
    \begin{split}
\hat q_{t+1}&=\Prob{\state{v_0}{t+1}=\red}=(1-p)\Prob{|R_{v_0}^{(t)}|\geq|B_{v_0}^{(t)}|}=(1-p)\Prob{\,\sum_{v\in S_{v_0}^{(t)}} \mathbf{1}_{\state{v}{t}=\red}\geq\frac{k+1}{2}}\,.
\end{split}
\end{equation}
Note that, for any siblings $v,w\in V$, the events ``$\state{v}{t}=\red$'' and ``$\state{w}{t}=\red$'' are independent at any round $t\in\mathbb{N}_0$. 
Hence the random variables in the family
$\{\mathbf{1}_{\state{v}{t}=\red} : v\in V\}$ are i.i.d.\ and $\Prob{\state{v}{t}=\red}=q_t$. 
So we have
\[
    \Prob{\,\sum_{v\in S_{v_0}^{(t)}} \mathbf{1}_{\state{v}{t}=\red}\geq\frac{k+1}{2}}=\Prob{\bin{k}{\hat q_t}\geq \frac{k+1}{2}}
\]
and hence by \cref{eq:din2-1} we get
\[
\hat q_{t+1}=(1-p)\Prob{\bin{k}{\hat q_t}\geq \frac{k+1}{2}}\,.
\]
By \cref{def:hatfpk} the above sequence can be rewritten as 
\[
    \hat q_{t+1}=
    \begin{cases}
    \widehat{F}_{p,k}(q) &\text{if }t= 0\,,
    \\\widehat{F}_{p,k}(\hat q_t) &\text{if }t\geq 1\,.
    \end{cases}
\]
Hence, the limit behavior in \cref{eq:th-tree2} and the monotonicity of the sequence $(\hat q_t)_t$ follow from \cref{lem:fixed_points_F-prel2}.
\end{proof}

\section{Phase Transition}\label{sec:phase_transition}
In this section we exploit the results in \cref{sec:preli} to analyze the behavior of \majority{k}{p}{\blue} and \modmajority{k}{p}{\blue} on sufficiently dense graphs.

We start by setting the ground for \cref{th:phasetransition-initial}, introducing the required notation.
For every node $u \in V$ and for every round $t$ we define the fraction of neighbors of $u$ in state \red as 
\(
    \rfrac{u}{t} \dfn \frac{| N_u\cap \R{t}|}{\deg{u}},
\)
where $\R{t}$ is defined in \cref{rtbt}.
Similarly we let $\rfracmax{t} \dfn \max_{u \in V} \rfrac{u}{t}$ denote the maximum fraction of neighbors in state \red at round $t$ over the nodes.

Given any configuration $\conf{t}=\bar{\mathbf{x}}$, we have that, for every $u \in V$, the expected fraction of neighbors of $u$ in state \red{} at round $t+1$ is as follows:
\begin{itemize}
    \item For the \majority{k}{p}{\blue}: 
    \begin{equation}\label{eq:exp_red}
        \Ex{\rfrac{u}{t+1} \condconf}
        = \frac{1}{\deg{u}} \sum_{v \in \neigh{u}} \Prob{\card{\bR{v}{t}} \geq \frac{k+1}{2} \condconf}
        = \frac{1}{\deg{u}} \sum_{v\in \neigh{u}} \Fpk{p}{k}(\rfrac{v}{t})\,,
    \end{equation}
    where $\bar R_v^{(t)}$ is defined in \cref{barsets}.

    \item For the \modmajority{k}{p}{\blue}:
    \begin{equation}\label{eq:exp_red2}
        \Ex{\rfrac{u}{t+1} \condconf}
        = \frac{1}{\deg{u}} \sum_{v \in \neigh{u}} (1-p)\Prob{\card{\R{t}_v} \geq \frac{k+1}{2} \condconf}
        = \frac{1}{\deg{u}} \sum_{v\in \neigh{u}} \widehat{F}_{p,k}(\rfrac{v}{t})\,.
    \end{equation}
\end{itemize}
Note that in \cref{eq:exp_red} we have applied \cref{def:fpk}, and used the fact that, given $\conf{t}$, we have $\card{\bR{v}{t}} \overset{d}{=} \bin{k}{(1-p)\rfrac{v}{t}}$, for every $v \in V$. Similarly, in \cref{eq:exp_red2} we have applied the definition of $\widehat{F}_{p,k}$ in \cref{def:hatfpk} and used the fact that, given $\conf{t}$, we have $\card{\R{t}_v} \overset{d}{=} \bin{k}{\rfrac{v}{t}}$, for every $v \in V$. 

\begin{remark}\label{rem:equivalence2}
In what follows, the only features of the dynamics that will be relevant are essentially those of the functions $\Fpk{p}{k}$ (for the \majority{k}{p}{\blue}) and $\widehat{F}_{p,k}$ (for the \modmajority{k}{p}{\blue}).
As seen in \cref{eq:equivalence} the graph of $\widehat{F}_{p,k}$ can be obtained from the graph of $\Fpk{p}{k}$ through contractions on both the coordinate axes. Hence in the rest of this section we will focus only on the \majority{k}{p}{\blue} and all the results presented in this section can be obtained for the \modmajority{k}{p}{\blue} by replacing $\Fpk{p}{k}$ by $\widehat{F}_{p,k}$ and $\varphi_{p,k}^-,\varphi_{p,k}^+$ by $\hat\varphi_{p,k}^-,\hat\varphi_{p,k}^+$, respectively.
\end{remark}

In order to state the main result of the paper, we need the following definition.
\begin{definition}\label{def:crit-pkq}
Fixed $q\in(\frac{1}{2},1]$, $k\geq 3$ and $p<\pk$, we define $\pkq$ as the unique solution of the optimization problem
\begin{equation}\label{pkq}
\pkq\coloneqq \max\left\{ p\in[0,\pk] \:\big\rvert\: \varphi_{p,k}^-\le q\right\}\,,
\end{equation}
where $\varphi_{p,k}^-$ and $\pk$ are defined as in \cref{lem:fixed_points_F-prel}. Note that $p_{k,1}^\star=p_k^\star$.
\end{definition}
By \cref{lem:phi-} we have that the map $p\in[0,\pk]\mapsto\varphi_{p,k}^-$ is
 increasing. Moreover, being composition of continuous functions, such a map is also continuous. Since $\varphi^-_{0,k}=\frac{1}{2}< q$, there exists some $\delta>0$ for which $\varphi^-_{p,k}\leq q$ for all $p\in[0,\delta]$. Hence $\pkq$ in \cref{pkq} is well defined. Moreover $\pkq$ is also increasing in $q$.
\begin{remark}\label{rmk:pkq compare}
 Note that, by \cref{def:crit-pkq}, the limit behavior of the sequence $(q_t)_t$ in \cref{th:tree} can be rewritten as 
 \begin{equation*}
\lim_{t\to\infty}q_t=
\begin{cases}
\varphi^+_{p,k}&\text{if }p<p_{k,q}^\star\,,
\\
0&\text{if }p>p_{k,q}^\star\,.
\end{cases}
\end{equation*}
Defining also $\hat p_{k,q}^\star$ as the analogue of $\pkq$ for the \modmajority{k}{p}{\blue}, that is
\begin{equation}\label{def:hatcrit-pkq}
    \hat p_{k,q}^\star:=\max\{p\in[0,\pk]\,|\,\hat\varphi^-_{p,k}\leq q\}\,,
\end{equation}
where $\hat\varphi^-_{k,q}$ is defined as in \cref{lem:fixed_points_F-prel2}, we have that the limit behavior of the sequence $(\hat q_t)_t$ in \cref{th:tree2} can be rewritten as
\begin{equation*}
\lim_{t\to\infty}\hat q_t=
\begin{cases}
\hat\varphi^+_{p,k}&\text{if }p<\hat p_{k,q}^\star\,,
\\
0&\text{if }p>\hat p_{k,q}^\star\,.
\end{cases}
\end{equation*}
Note that, by \cref{eq:equivalence}, \cref{lem:fixed_points_F-prel} and \cref{lem:fixed_points_F-prel2}, we get that 
$\hat\varphi_{p,k}^+=(1-p)\varphi_{p,k}^+$ and hence, by calling $s=\min\left\{1,\,\frac{q}{1-p}\right\}$, we get
\[
    \modpkq=p_{k,s}^\star\geq p_{k,q}^\star\,.
\]
Moreover, the equality $\modpkq=p_{k,q}^\star$ holds if and only if $q=\frac{1}{2}$ (where both take value $0$) or if $q\ge \varphi_{\pk,k}$ (defined in \cref{lem:fixed_points_F-prel}). For the first implication it is enough to realize that
if $q\ge \varphi_{p^\star_k,k}$ then, by the monotonicity of $p\mapsto\varphi^-_{p,k}$, we have
\[q\ge \varphi_{p^\star_k,k}>\varphi^-_{p,k}>\hat\varphi^-_{p,k} \] 
and by \cref{pkq} and \cref{def:hatcrit-pkq} it immediately follows that 
\[
    \modpkq=p_{k,q}^\star=\pk\,. 
\]
For the other implication, notice that if there exists some $r$ such that
\[
    r=\max\left\{p\in[0,\pk] \cond \varphi^-_{p,k}\leq \frac{q}{1-p}\right\}= \max\left\{p\in[0,\pk] \cond \varphi^-_{p,k}\leq q\right\}\,,
\]
then it must be that $r\in \{0,\pk\}$. 
Moreover, since the solutions of \cref{pkq} and \cref{def:hatcrit-pkq} are monotonically increasing and $\hat p^\star_{k,q} \ge p^\star_{k,q}$, in order to have the equality it is necessary that $q$ is sufficiently big to imply $p^\star_{k,q}=\pk$. The conclusion follows by noting that
\[
    \inf\left\{q\ge 1/2 \cond p^\star_{k,q}=\pk \right\} =\varphi_{\pk,k}\,.  
\]
\end{remark}

The main result of this work is described rigorously by the forthcoming \cref{th:phasetransition-initial}, which shows the existence of a \emph{phase transition} in the bias $p$ when the nodes of a dense graph execute \majority{k}{p}{\blue} starting with a random configuration in which each vertex is in state \red independently with probability $q$.
Roughly speaking, on the one hand \cref{th:phasetransition-initial} states that, if the bias is smaller than the \emph{critical value} $\pkq$, a superpolynomial number of rounds is needed to reach disruption:
we will show that the system will remain trapped in a \emph{metastable phase} in which the volume of nodes in state $\red$ is some constant fraction of the total (larger than $1/2$) for every polynomial number of rounds. 
On the other hand, if the bias is larger than $\pkq$, a constant number of rounds w.h.p.\ suffices to reach disruption.

\begin{theorem}\label{th:phasetransition-initial}
Fix $p\in[0,1]$ and consider a sequence of graphs $(G_n)_{n\in\mathbb{N}}$ such that $\min_{v\in V}\deg{v}=\omega(\log n)$. For every fixed $n$ consider the \majority{k}{p}{\blue} dynamics with odd $k\geq 3$. Assume that the system starts from an initial configuration $\conf{0}$ in which each vertex is $\red$ with probability $q\in(1/2,1]$, independently of the others. 
\begin{enumerate}
 \item \emph{Slow disruption:} if $p<\pkq$, then for all constant $\gamma>0$ there exists a constant $T=T(p,k,q,\gamma)$ such that for all constant $K>0$
    \[
        \Prob{\forall t\in[T,n^K],\,\,
        \frac{|R^{(t)}|}{n}\ge \varphi^+_{p,k}-\gamma}
        =1-o(1)\,,
    \]
    where $\varphi^+_{p,k}$ is defined as in \cref{lem:fixed_points_F-prel}. Hence, as a corollary,
     for every constant $K>0$
        \[
            \Prob{\tau\ge n^K} =1-o(1)\,.
        \]
    
    \item \emph{Fast disruption:} if $p>\pkq$, then  there exists a constant $T=T(p,k)$ s.t.
    \begin{equation}\label{eq:413}
        \Prob{\tau\le T}=1-o(1)\,.
    \end{equation}
\end{enumerate}
\end{theorem}
\begin{proof}
The proof for the slow disruption regime immediately follows from \cref{coro:prop1,coro:prop2}, which are stated and proved in \cref{sec:slow disruption}.

The proof for the fast disruption regime, instead, follows from \cref{prop:supercritical} (by choosing $\gamma<1/2$), stated and proved in \cref{sec:fast disruption}, and from \cref{coro:prop1} (where $q_t=0$).
\end{proof}

\subsection{Slow Disruption}\label{sec:slow disruption}
The proof for the slow disruption regime is based on a concentration result of the \red density in large graph dynamics around the quantities appearing in the \emph{tree-process} described in \cref{sec:preli}. 
More precisely, \cref{prop:subcritical-general,prop:subcritical2} show that for every round $t\in\text{poly}(n)$ all the nodes in the graph have a fraction of neighbors in state \red which is asymptotically equal to $q_t$, as defined in \cref{eq:def-q}.
That is to say that the \emph{tree-process} is a good uniform approximation of the actual graph process for every polynomial number of rounds. We start with \cref{prop:subcritical-general} which considers every starting configuration with $q>\mu_{p,k}$ (where $\mu_{p,k}$ is the unique point in $[\varphi^-_{p,k},1]$ for which $\Fpk{p}{k}^{\,'}(\mu_{p,k})=1$, see \cref{fig:qualplot}). 
For such values of the triple $(p,k,q)$ we can rely on the contraction property of $\Fpk{p}{k}$.

\begin{proposition}\label{prop:subcritical-general}
Fix $q\in(1/2,1]$ and $p<\pkq$. Consider the \majority{k}{p}{\blue} starting from a random initial configuration in which each vertex is $\red$ independently of the others, with probability $q$. Define $\mu_{p,k}$ as in \cref{claim:L}. If $q\in(\mu_{p,k},1]$, then, for all constants $\gamma,K>0$, we have
\[
    \Prob{\forall t\le n^K,\,\forall v\in V,
    \, \phi_v^{(t)}\in [q_t-\gamma, q_t+\gamma]}=1-o(1)\,,
\]
where the sequence $(q_t)_{t\geq 0}$ is defined recursively as in \cref{eq:def-q}.
\end{proposition}
\begin{proof}
Fix $\gamma>0$ arbitrarily small.	
For all $t\geq 0$ and $v\in V$ consider the events
	\begin{equation}\label{eq:def-et}
	    \event{v}{t}=\left\{\phi_v^{(t)}\in[q_t-\gamma,q_t+\gamma]\right\}\,,
	    \qquad 
	    \event{}{t}=\bigcap_{v\in V}\event{v}{t}\,.
   \end{equation}
	We start at $t=0$ by computing the probability that the event $\event{}{0}$ occurs, i.e.,
	\[
	    \Prob{\event{}{0}}=\Prob{\forall v\in V,\: \rfrac{u}{0}\in[q-\gamma,q+\gamma]}\,.
	\]
	Fixed any $v\in V$, we look for a lower bound for the probability $\event{v}{0}$ which holds uniformly in $v\in V$.
	\begin{align*}
	\Prob{\event{v}{0}}
	&=\Prob{\frac{\sum_{w\in \neigh{v}}\mathbf{1}_{w\in \R{0}}}{\deg{v}}\in[q-\gamma,q+\gamma]}
	=\Prob{\left|\sum_{w\in \neigh{v}}\mathbf{1}_{w\in \R{0}}-q\deg{v}  \right|\le \gamma \deg{v} }
	\\
	&=\Prob{\left|\bin{\deg{v}}{q} - q\deg{v} \right|\le \gamma\deg{v} }
	= 1-e^{-\omega(\log n)}\,,
	\end{align*}
	where in the third equality we used the fact that
	\[
	    \Prob{w\in \R{0}}=q\,,
	    \qquad
	    \forall w\in V\,,
	\]
	and that the events $\left\{w\in \R{0} \right\}$ and  $\left\{u\in \R{0} \right\}$ are independent for $u\neq w$. Moreover, in the last asymptotic equality we used the classical multiplicative version of Chernoff's bound (\cref{thm:chernoff extension}). By the \emph{density assumption} $\min_{v\in V}\deg{v}=\omega(\log n),$ the lower bound above holds uniformly in $v\in V$ and hence, by the union bound, we get
	\begin{equation}\label{eq:eq}
	    \Prob{\event{}{0}}=1-n\cdot e^{-\omega(\log n)}\,.
	\end{equation}
	It is worth noting that up to this point we did not use the assumption on $p$, hence \cref{eq:eq} holds also in the case $p\ge \pkq$.
	
	Now assume that the configuration $\conf{t}$ is such that $\conf{t}\in \event{}{t}$. Fixed any $v\in V$ we look for a lower bound for the conditional probability $\Prob{\event{v}{t+1}\mid \conf{t}}$, which holds uniformly in $v\in V$ and in $\conf{t}\in \event{}{t}$. We start by rewriting
	\[
	    \Prob{\event{v}{t+1}\mid \conf{t}}=\Prob{\left|\sum_{w\in \neigh{v}}\mathbf{1}_{w\in \R{t+1}}-q_{t+1}\deg{v}   \right|\le \gamma\deg{v} \:\bigg\rvert\: \conf{t} }\,.
	\]
	Note that under any $\conf{t}\in  \event{}{t}$ the following stochastic domination holds
	\[
	    \bin{\deg{v}}{\Fpk{p}{k}(q_t-\gamma)}\preceq 
	    \sum_{w\in \neigh{v}}\mathbf{1}_{w\in \R{t+1}}\big\rvert\: \conf{t} \preceq \bin{\deg{v}}{\Fpk{p}{k}(q_t+\gamma)}\,.
    \]
	Hence, called
	\[
	    X\sim \bin{\deg{v}}{\Fpk{p}{k}(q_t-\gamma)}\,,\qquad Y\sim \bin{\deg{v}}{\Fpk{p}{k}(q_t+\gamma)}\,,
	\]
	it is sufficient to show that, for $Z=X,Y$, it holds that
	\begin{equation}\label{eq:binomial}
	    \Prob{\left|Z-q_{t+1}\deg{v} \right|\le\gamma \deg{v}}=1-e^{-\omega(\log n)}\,.
	\end{equation}
	We show \cref{eq:binomial} for the case $Z=X$, being the proof for the case $Z=Y$ identical. By the triangle inequality and the fact that $q_{t+1}=\Fpk{p}{k}(q_t)$ we have
	\begin{equation}\label{eq:binomialX}
	    \Prob{\left|X-q_{t+1}\deg{v}   \right|\le \gamma\deg{v}}\ge \Prob{\left|X-\Fpk{p}{k}(q_t-\gamma)\deg{v}   \right| + \left|\Fpk{p}{k}(q_t-\gamma)-\Fpk{p}{k}(q_t) \right|\deg{v}\le \gamma\deg{v} }\,.
	\end{equation}
	By \cref{claim:L}, 
	for all $\eta>0$ and $x,y\in[\mu_{p,k}+\eta,1]$ there exists some $L=L(\eta)<1$ such that
	\[
	    |\Fpk{p}{k}(x)-\Fpk{p}{k}(y)|\le L |x-y|\,,
	\]
	and in particular
\begin{equation}\label{eqlast2}
     |\Fpk{p}{k}(q_t)-\Fpk{p}{k}(q_t-\gamma)|\le L \gamma\,.
\end{equation}
Hence, by \cref{eq:binomialX,eqlast2}, we infer
	\[
	     \Prob{\left|X-q_{t+1}\delta_v   \right|\le \gamma\delta_v}\ge\Prob{\left|X-\Fpk{p}{k}(q_t-\gamma) \right|\le (1-L)\gamma \delta_v}=1-e^{-\omega(\log n)}\,.
	\]
	Repeating the same argument with $Y$ instead of $X$, we deduce that \cref{eq:binomial} holds. Hence, by the uniformity of the argument in $v\in V$ and in $\conf{t}\in \event{}{t}$, and by a union bound over $v\in V$, we get that for any $t\geq 0$
	\begin{equation}\label{eq:et}
	    \Prob{\event{}{t+1}\mid \event{}{t}}=1-n\cdot e^{-\omega(\log n)}\,.
	\end{equation}
 So by \cref{eq:eq} and \cref{eq:et}, for all $K>0$ we have
	\begin{align*}
	   \Prob{\forall v\in V,\:\forall t\in[0,n^K],\:\rfrac{v}{t}\in[q_t-\gamma,q_t+\gamma]}&= \Prob{\event{}{0}}\prod_{s=1}^{n^K}\Prob{\event{}{s}\cond \event{}{s-1}}\\
	   &\ge 1-n^{K+1}\cdot e^{-\omega(\log n)}\,.
	\qedhere
	\end{align*}
\end{proof}

 \begin{corollary}\label{coro:prop1}
 In the same setting of \cref{prop:subcritical-general} and with the same choice of the constants, it holds
 \[\Prob{\forall t\le n^K, \frac{|\R{t}|}{n}\in[q_t-\gamma,q_t+\gamma]}=1-o(1)\,.
 \]
 \end{corollary}
 \begin{proof}
 Fix $t\ge0$ and consider assume that the configuration $\conf{t}$ is such that  $\conf{t}\in\mathcal{E}_t$, where $\mathcal{E}_t$ is  defined in \cref{eq:def-et}. We have
 \[
    \sum_{v\in V}\mathbf{1}_{v\in \R{t+1}}\,| \,\conf{t} \overset{\rm d}{=}\bigoplus_{v\in V}{\rm Bern}(F_{p,k}(\phi_v^{(t)}))\,,
 \]
 where the symbol $\bigoplus$ denotes the sum of independent variables, while the symbol $\overset{\rm d}{=}$ denotes that the random variables in both the members of the identity have the same distribution. Notice that, by definition of $\mathcal{E}^{(t)}$ we have
\[
 \Ex{\sum_{v\in V}\mathbf{1}_{v\in \R{t+1}}\,\mid\,\conf{t}}\in [q_{t+1}-\gamma,q_{t+1}+\gamma]\,.
 \]
 Therefore, by an immediate application of Hoeffding's inequality (\cref{thm:hoeffding}) we get
 \begin{equation}\label{eq:azuma}
 \Prob{ \frac{|\R{t+1}|}{n} \in [q_{t+1}-2\gamma,q_{t+1}+2\gamma]  \condconf}= 1-e^{-\Omega(n)}\,.
 \end{equation}

In conclusion, called 
 \[ 
 \mathcal{F}^{(t)} = \left\{ \frac{|\R{t}|}{n} \not\in [q_t-2\gamma,q_t+2\gamma]\right\},\qquad\forall t\ge 0\,,
 \]
and recalling the definition of $\mathcal{E}^{(t)}$ in \cref{eq:def-et} and using \cref{prop:subcritical-general} we deduce that, for all $K>0$,
\begin{align*}
    \Prob{\exists t\le n^K\,, \frac{|\R{t}|}{n}\not\in[q_t-2\gamma,q_t+2\gamma]}
    &\le n^K \max_{t\le n^K} \Prob{\mathcal{F}^{(t)}} 
    \\
    &= n^K \max_{t\le n^K} \Prob{\mathcal{E}^{(t)}}\Prob{\mathcal{F}^{(t)}\cond\mathcal{E}^{(t)}}+e^{-\omega(\log n)}
    \\
    &= o(1)\,,
\end{align*}
where in the last inequality we used \cref{eq:azuma}.
\end{proof}
The forthcoming \cref{prop:subcritical2} closes the gap where $q\in(\varphi^-_{p,k},\mu_{p,k}]$. Informally, \cref{prop:subcritical2} claims that within a constant number of rounds every node has at least a fraction $\mu_{p,k}$ of neighbors in state \red. Moreover, the fractions of neighbors in state \red of every vertex coincide at the first order.
\begin{proposition}
\label{prop:subcritical2}
 Fix $q\in(1/2,1]$ and $p<\pkq$. Consider the \majority{k}{p}{\blue} starting from the initial configuration in which each vertex is $\red$ independently of the others, with probability $q$. If $q\in(\varphi^-_{p,k},\mu_{p,k}]$, then there exists some $T=T(p,k,q)$ such that
\(
    q_{T}>\mu_{p,k}.
\)
Moreover, for all constant $\gamma>0$ 
\[
\Prob{\forall t\le T,\:\forall v\in V,\:\rfrac{v}{t}\in[q_{t}-\gamma,q_{t}+\gamma]}=1-e^{-\omega(\log n)}\,.
\]
\end{proposition}
\begin{proof}
We start the proof by mimicking the proof of \cref{prop:subcritical-general}. Consider the sequence
\begin{equation*}
    \rho_t:=(2L)^t\rho_0\,,
\end{equation*}
where
\begin{equation*}
L:=\max_{x\in[\varphi^-_{p,k},1]}\Fpk{p}{k}'(x)\ge 1
\end{equation*}
and $\rho_0$ will be defined later. For all $t\geq 0$ and $v\in V$ consider the events
	\[
	    \event{v}{t}=\left\{\phi_v^{(t)}\in[q_t-\rho_t,q_t+\rho_t]\right\}\,,
	    \qquad 
	    \event{}{t}=\bigcap_{v\in V}\event{v}{t}\,.
	\]
	We start at $t=0$ by computing the probability that the event $\event{}{0}$ occurs. Such a probability is obtained directly by \cref{eq:eq} with $\rho_0$ instead of $\gamma$, i.e.,
	\begin{equation}\label{eq:E0}
	    \Prob{\event{}{0}}=\Prob{\forall u\in V,\, \rfrac{u}{0}\in[q-\rho_0,q+\rho_0]}=1-e^{\omega(\log n)}\,.
	\end{equation}
	As in the proof of \cref{prop:subcritical-general}, assume that the configuration $\conf{t}$ is such that  $\conf{t}\in \event{}{t}$. Fixed any $v\in V$ we look for a lower bound for the conditional probability $\Prob{\event{v}{t+1}\mid \conf{t}}$, which holds uniformly in $v\in V$ and in $\conf{t}\in \event{}{t}$. We start by rewriting
	\[
	    \Prob{\event{v}{t+1}\mid\, \conf{t}}=\Prob{\left|\sum_{w\in \neigh{v}}\mathbf{1}_{w\in \R{t+1}}-q_{t+1}\deg{v}   \right|\le \rho_{t+1}\deg{v} \:\bigg\rvert\: \conf{t} }\,.
	\]
	Note that under any $\conf{t}\in  \event{}{t}$ the following stochastic domination holds
	\[
	    \bin{\deg{v}}{\Fpk{p}{k}(q_t-\rho_t)}\preceq 
	    \sum_{w\in \neigh{v}}\mathbf{1}_{w\in \R{t+1}}\big\rvert\: \conf{t} \preceq \bin{\deg{v}}{\Fpk{p}{k}(q_t+\rho_t)}\,.
    \]
	Hence, called
	\[
	    X\sim \bin{\deg{v}}{\Fpk{p}{k}(q_t-\rho_t)}\,,
	    \qquad 
	    Y\sim \bin{\deg{v}}{\Fpk{p}{k}(q_t+\rho_t)}\,,
	\]
	it is sufficient to show that, for $Z=X,Y$, it holds that
	\begin{equation}\label{eq:binomial2}
	\Prob{\left|Z-q_{t+1}\deg{v} \right|\le\rho_{t+1} \deg{v}}=1-e^{-\omega(\log n)}\,.
	\end{equation}
	We show \cref{eq:binomial2} for the case $Z=X$, with the proof for the case $Z=Y$ being identical. By the triangle inequality and the fact that $q_{t+1}=\Fpk{p}{k}(q_t)$ we have
	\begin{align*}
	\Prob{\left|X-q_{t+1}\deg{v}   \right|\le \rho_{t+1}\deg{v}}
	&\ge \Prob{\left|X-\Fpk{p}{k}(q_t-\rho_t)\deg{v}   \right| + \left|\Fpk{p}{k}(q_t-\rho_t)-\Fpk{p}{k}(q_t) \right|\deg{v}\le \rho_{t+1}\deg{v} }
	\\
	&\ge \Prob{|X-\Fpk{p}{k}(q_t-\rho_t)\deg{v}|\le \big(\rho_{t+1}-L\rho_t\big)\deg{v}}
	\\
	&=\Prob{|X-\Fpk{p}{k}(q_t-\rho_t)\deg{v}|\le L\rho_t\deg{v}}
	\\
	&\ge 1-e^{-\omega(\log n)}\,.
	\end{align*}
Therefore by the uniformity of the argument in $v\in V$ and in $\conf{t}\in \event{}{t}$, and by a union bound over $v\in V$, we get that for any $t\geq 0$
\begin{equation}\label{eq:et2}
    \Prob{\event{}{t+1}\mid \event{}{t}}=1-n\cdot e^{-\omega(\log n)}\,.
\end{equation}
Notice that, if $q_0\in(\varphi^-_{p,k},\mu_{p,k})$, by \cref{th:tree} $q_t$ converges to $ \varphi^+_{p,k}$ from below and hence there exists some finite
\[
    T=T(p,k,q)      \coloneqq\inf\left\{ t\ge 0\::\:q_t>\mu_{p,k}\right\}<+\infty\,.
\]
Define the sequence $(g_t)_{t\in\mathbb{N}}$, where
\[
    g_t\coloneqq\Fpk{p}{k}(q_{t-1}-\rho_{t-1})\,.
\]
We prove that
 $g_t$  is increasing and its increments are lower bounded  by a constant uniformly in $t\le T$.
Let us start by noting that
\begin{align*}
g_{t+1}-g_{t} &= \Fpk{p}{k}(q_t-\rho_t)-\Fpk{p}{k}(q_{t-1}-\rho_{t-1})
\\
&= \Fpk{p}{k}(q_t-\rho_t)-\Fpk{p}{k}(q_t)+\Fpk{p}{k}(q_t)-\Fpk{p}{k}(q_{t-1})+\Fpk{p}{k}(q_{t-1})-\Fpk{p}{k}(q_{t-1}-\rho_{t-1})
\\
&> \Fpk{p}{k}(q_t-\rho_t)-\Fpk{p}{k}(q_t)+\Fpk{p}{k}(q_t)-\Fpk{p}{k}(q_{t-1})\,.
\end{align*}
Hence it is sufficient to show that
\begin{equation}\label{eq43}
\Fpk{p}{k}(q_t)-\Fpk{p}{k}(q_t-\rho_t)< \Fpk{p}{k}(q_{t})-\Fpk{p}{k}(q_{t-1})\,.
\end{equation}
Define 
\[
    C=\inf_{t\in[0,T]}\big(\Fpk{p}{k}(q_{t})-\Fpk{p}{k}(q_{t-1})\big) >0\,.
\]
Then, we are left to show that the left hand side of \cref{eq43} is strictly smaller than $C$. This can be done by noting
\[
    \Fpk{p}{k}(q_t)-\Fpk{p}{k}(q_t-\rho_t)\le L\rho_t=L\cdot(2L)^t\rho_0\le (2L)^{T+1}\rho_0
\]
and defining, for any given $\gamma\in(0,1)$
\[
    \rho_0\coloneqq\gamma\cdot \frac{C}{(2L)^{T+1}}\,.
\]
In conclusion, by \cref{eq:E0} and \cref{eq:et2}, for any $\gamma\in(0,1)$  we can deduce that 
\begin{align*}\label{eq:luna}
    \Prob{\forall v\in V,\:\forall t\le T,\:\phi_v^{(t)}\in[q_t-\gamma,q_{t}+\gamma]}
    &\ge \Prob{\event{}{0}}\prod_{s=1}^{T}\Prob{\event{}{s}\cond\event{}{s-1}}
    \\
    &\ge 1-\bigO(n^{-K})\,,
    \qquad
    \forall K>0\,.\qedhere
\end{align*}
\end{proof}

Using the same argument as in \cref{coro:prop1}, we deduce the same concentration result of \cref{prop:subcritical2} for the number of vertices in state \red.
\begin{corollary}\label{coro:prop2}
 In the same setting of \cref{prop:subcritical2} and with the same choice of the constants, it holds
 \[\Prob{\forall t\le T, \frac{|\R{t}|}{n}\in[q_t-\gamma,q_t+\gamma]}=1-e^{-\omega(\log n)}\,.
 \]
\end{corollary}

\subsection{Fast Disruption}\label{sec:fast disruption}

In the fast disruption regime, for every node of the graph, we can upper bound the expected fraction of neighbors in state \red. This fact is formalized by \cref{lem:supercritical-text} (if $p>\pk$) and \cref{claim-lemma} (if $p\in(\pkq,\pk]$).
In both the cases we will essentially prove that the fraction of neighbors of $u$ in state \red, maximized over $u\in V$, is a supermartingale.
\begin{lemma}\label{lem:supercritical-text}
Fix $p>\pk$. Then there exists some $\varepsilon=\varepsilon(p,k)>0$ such that for all $u \in V$, $t\ge 0$, and every configuration $\bar{\mathbf{x}}$ it holds that
\begin{equation*}
  \Ex{\rfrac{u}{t+1} \condconf} 
    \leq (1-\varepsilon) \rfracmax{t}\,.
\end{equation*}
\end{lemma}
\begin{proof}
    Recall the formulation, given in \cref{eq:exp_red}, of the expected fraction of neighbors of a node $u$ that are in state \red{} at round $t+1$.
	Being $p>\pk$, by \cref{lem:fixed_points_F-prel}, it holds $\Fpk{p}{k}(x)<x$ for all $x\in[0,1]$. Moreover, the following claim holds true.
	\begin{claim}\label{claim:phipk}
	Let $p>p_k^\star$.
	There exists a point $z_{p,k}\in[0,1]$ s.t.\ the linear function 
	$r:[0,1]\to\mathbb{R}$ defined as
	\[
	    r(x)=\frac{1-p}{1-\pk}\cdot x 
	\]
	is such that
	$x>r(x)\ge \Fpk{p}{k}(x)$ for any $x\in(0,1]$.
	Moreover, $r(x)=\Fpk{p}{k}(x)$ only for $x=0$ and $x=z_{p,k}$, where $z_{p,k}\in\left[\min\left\{\frac{1}{2(1-p)},1\right\},1\right]$.
	\end{claim}
	\begin{proof}[Proof of \cref{claim:phipk}]
	Suppose that $p\geq\frac{1}{2}$. Then $F_{p,k}$ is a convex function in $[0,1]$. Hence 
	\[
	    F_{p,k}(x)\leq (F_{p,k}(1)-F_{p,k}(0))x+F_{p,k}(0)=F_{p,k}(1)\cdot x\,.
	\]
	So, defining $z_{p,k}:=1$, we get the thesis in the case $p\geq \frac{1}{2}$.	Suppose now that $p<\frac{1}{2}$.
    	Imposing the equality
    	\[
    	    \Fpk{p}{k}(x)=r(x)\qquad\iff\qquad\Prob{\text{Bin}(k,(1-p)x)\ge\frac{k+1}{2}}=\frac{1-p}{1-\pk}x\,,
    	\]
    	it is immediate that a solution is given by $x=0$. On the other hand, called 
    	$y=\frac{1-p}{1-\pk}\cdot x $,
    	the above equation reads
    	\[
    	    \Fpk{p}{k}\left(\frac{1-\pk}{1-p}\cdot y \right)=y\,.
    	\]
    	The latter can be rewritten as
    	\[
    	    \Prob{\text{Bin}\left(k,(1-\pk)y \right)\ge \frac{k+1}{2}}=y\qquad\iff\qquad \Fpk{\pk}{k}(y)=y\,,
    	\]
    	which we know from \cref{lem:fixed_points_F-prel} to have only a non-trivial solution different from zero. Called $\bar y\neq 0$ such a solution, we have that
    	\[
    	    z_{p,k}:=\frac{1-\pk}{1-p}\cdot\bar y\,.
    	\]
    	Note that $\bar{y}>\frac{1}{2(1-p)}$ and hence, since $\frac{1-\pk}{1-p}>1$, we have that $z_{p,k}>\frac{1}{2(1-p)}$.
	\end{proof}
	By \cref{claim:phipk} we have that
	\begin{equation}\label{eq:Fpk_less}
    	\Fpk{p}{k}(\rfrac{v}{t}) < 
    	\frac{1-p}{1-\pk}\rfrac{v}{t}
    	\leq (1 - \varepsilon)\rfrac{v}{t}
	\end{equation}
	where 
	\[
	    \varepsilon =\varepsilon(p,k):=
	    \frac{1}{2}\cdot\frac{p-\pk}{1-\pk}\,.
	\]
		Therefore, we can conclude the proof of \cref{lem:supercritical-text} just by combining \cref{eq:Fpk_less,eq:exp_red}, getting
	\begin{equation*}
	    \Ex{\rfrac{u}{t+1} \condconf} 
    	= \frac{1}{\deg{u}} \sum_{v\in \neigh{u}} \Fpk{p}{k}(\rfrac{v}{t})
    	\leq \frac{1}{\deg{u}} \sum_{v\in \neigh{u}} (1 - \varepsilon) \rfrac{v}{t}
    	\leq (1 - \varepsilon) \rfracmax{t}\,.
    	\qedhere  
	\end{equation*}
\end{proof}
\begin{lemma}
\label{claim-lemma}
Let $p\in(p_{k,q}^\star,p_k^\star]$ and consider any $\conf{t}$ such that 
\(
    \rfracmax{t}\le \varphi^-_{p,k}-\eta,
\) for some $\eta>0$.
Then, there exists some constant $\varepsilon=\varepsilon(p,k,\eta)>0$ such that
\begin{equation*}
    \Ex{\rfrac{u}{t+1}\condconf}\le (1-\varepsilon)\rfracmax{t}\,.
\end{equation*}
\end{lemma}
\begin{proof}
Let $x\in[0,\varphi^-_{p,k}-\eta]$ and consider the line
\[
    r(x):=\frac{\Fpk{p}{k}(\varphi^-_{p,k}-\eta)}{\varphi^-_{p,k}-\eta}x\,.
\]
By \cref{lem:fixed_points_F-prel} we have that $\Fpk{p}{k}(x)<r(x)<x$ for every $x\in(0,\varphi^-_{p,k}-\eta]$.
Hence, by defining
\[
    \varepsilon=\varepsilon(p,k,\eta):=\frac{1}{2}\left(1-\frac{\Fpk{p}{k}(\varphi^-_{p,k}-\eta)}{\varphi^-_{p,k}-\eta}\right)\,,
\]
we get $\Fpk{p}{k}(x)<(1-\varepsilon) x$ for every $x\in(0,\varphi^-_{p,k}-\eta]$.
We complete the proof of the claim by computing the conditional expectation, namely
\[
    \Ex{\rfrac{u}{t+1} \condconf}
    = \frac{1}{\deg{u}} \sum_{v\in \neigh{u}} \Fpk{p}{k}(\rfrac{v}{t})<\frac{1}{\deg{u}} \sum_{v\in \neigh{u}}(1-\varepsilon)\rfrac{v}{t}\le(1-\varepsilon)\rfracmax{t}\,.
    \qedhere
\]
\end{proof}
\noindent
Then, in order to conclude the proof of \cref{eq:413} we use the next proposition.
\begin{proposition}
\label{prop:supercritical}
Fix $p>\pkq$. For all $\gamma>0$ there exists some $T=T(p,k,q,\gamma)$ such that
\begin{equation*}
    \Prob{\exists t\le T \text{ s.t. }\,\rfracmax{t} \leq \gamma}
    =1-o(1)\,.
\end{equation*}
\end{proposition}
\begin{proof}
We divide the proof into two parts: $p>\pk$ and $p\in(\pkq,\pk]$. We start by assuming $p>\pk$. Consider any $u\in V$ and fix an arbitrary $\gamma>0$. Then, by \cref{lem:supercritical-text}, we immediately get 
\begin{equation}\label{eq:boundEp>pk}
       \Ex{|\R{t+1}_u| \mid \conf{t}} 
    \leq (1-\varepsilon)\rfracmax{t} \deg{u}\leq(1-\varepsilon)\max\{\rfracmax{t},\gamma\} \deg{u}\,. 
\end{equation}
We now aim at deriving an analogue bound as that in \cref{eq:boundEp>pk} in the alternative case in which $p\in(\pkq,\pk]$. Start by noting that, by the definition of $\pkq$ in \cref{pkq}, if $p\in(\pkq,\pk]$ it must hold $q<\varphi_{p,k}^-$. Therefore, thanks to \cref{eq:E0} the event $\rfracmax{0}<\varphi^-_{p,k}$ holds with high probability. Hence, we can apply \cref{claim-lemma} and conclude that \cref{eq:boundEp>pk} for every $p>\pkq$.
We aim at bounding the quantity $\rfrac{u}{t+1}$, namely
\begin{equation*}
   \begin{split}
    \Prob{\rfrac{u}{1} > (1-\varepsilon^2)\max\{\rfracmax{0},\gamma\} \mid \conf{0}}
    &= \Prob{|\R{1}_u| > (1-\varepsilon^2) \max\{\rfracmax{0},\gamma\}\deg{u} \mid \conf{0}}
    \\&
    = \Prob{|\R{1}_u| > (1+\varepsilon) (1-\varepsilon) \max\{\rfracmax{0},\gamma\} \deg{u} \mid \conf{0}}\,.
\end{split}
\end{equation*}
Note that, given the configuration at round $0$, $|\R{1}_u|$ is a Binomial random variable and, thus, by applying a multiplicative form of Chernoff's Bound (\cref{thm:chernoff extension}), we get
\[
    \Prob{\rfrac{u}{1} > (1-\varepsilon^2)\max\{\rfracmax{0},\gamma\} \mid \conf{0}}
    \leq \exp\left\{-\frac{\varepsilon^2(1-\varepsilon)}{3} \max\{\rfracmax{0},\gamma\} \deg{u}\right\}\,.
\]
Since $\deg{u} = \omega(\log n)$ by hypothesis and $\max\{\gamma,\rfracmax{0}\}>\gamma$ regardless of $\rfracmax{0}$, by taking a union bound over $u\in V$ and integrating over the conditioning we get,
\begin{equation}\label{eq:panco2}
    \Prob{\rfracmax{1} > (1-\varepsilon^2)\max\{\gamma,\rfracmax{0}\} }
    \leq n\cdot e^{-\omega(\log n)}\,.
\end{equation}
By iterating the same argument leading to \cref{eq:panco2} we obtain that, for all $t\ge 1$,
\begin{equation}\label{eq:panco2bis}
    \Prob{\rfracmax{t} > (1-\varepsilon^2)\max\{\gamma,\rfracmax{t-1}\} }
    \leq n\cdot t\cdot e^{-\omega(\log n)}\,.
\end{equation}
Let us call
\[
    T=T(p,k,\gamma)\coloneqq\min\{t\ge0\:|\: (1-\varepsilon^2)^t< \gamma \}\,, 
\]
and define the event
\(
    \mathcal{F}:=\left\{ 
        \forall t< T,\:\rfracmax{t+1} 
        \le (1-\varepsilon^2)\max\{\gamma,\rfracmax{t}\}
    \right\}\,.
\)
By \cref{eq:panco2bis} we have
$\Prob{\mathcal{F}}\ge 1-n\cdot T^2\cdot e^{-\omega(\log n)}$.
Called 
\(
    \mathcal{E}:=\left\{ 
        \exists t\le T \text{ s.t. } \rfracmax{t}\leq\gamma
    \right\}\supset\mathcal{F}\,,
\)
we get the thesis.
\end{proof}

\section{Limit Cases: Voter and Deterministic Majority Dynamics}\label{sec:limit}
In this section we analyze the two limit cases of \majority{k}{p}{\blue} and \modmajority{k}{p}{\blue}, 
considering the case $k=1$ as well as the case in which $k$ is large.
In particular, $1$-\maj{} is equivalent to the \voter{}, i.e., nodes copy the state of a randomly sampled neighbor. 
In \cref{ssec:voter} we analyze its behavior on the biased communication models described in \cref{sec:model}. Notice that, in the case of the biased \voter{}, there is no difference in placing the bias on the edges or on the nodes, i.e., \majority{1}{p}{\blue} and \modmajority{1}{p}{\blue} coincide. 
On the other hand, for large values of $k$, one might expect the behavior of \kmaj{} to be similar to that of \determ{} \maj{}, in which nodes update their state to that supported by the majority of nodes in their entire neighborhood. In \cref{ssec:det_maj} we investigate the relation between the two dynamics in our biased communication models.

\subsection{Voter Dynamics}\label{ssec:voter}
Differently from the general case of \cref{th:phasetransition-initial} where $k\geq3$, no phase transition is observed for \majority{1}{p}{\blue} and \modmajority{1}{p}{\blue}, due to the linearity of the dynamics.
Moreover, the effect of the bias $p$ has a strong impact on its behavior, compared to the case $p=0$. 
We show in \cref{thm:voter} that in the biased communication model disruption is reached in $\bigO(1)$ rounds, w.h.p., regardless of the initial configuration and of the parameter $p$.
In other words, in a constant number of rounds any majority on state $\red$ is subverted, w.h.p.
\begin{proposition}
\label{thm:voter}
    Fix $p\in(0,1]$ and consider a sequence of graphs $(G_n)_{n\in\mathbb{N}}$ such that $\min_{u\in V}\deg{u}=\omega(\log n)$.
    For any fixed $n$, consider the \majority{1}{p}{\blue} (equivalently, \modmajority{1}{p}{\blue}) dynamics  with any initial configuration $\conf{0}$. Then, there exists $T=T(p)$ such that
    \[
        \Prob{\tau\le T}=1-o(1)\,.
    \]
\end{proposition}
\begin{proof}
    Recall the definition of $\bar R_v^{(t)}$ in \cref{barsets}. Let us compute, for each node $u \in V$, an upper bound to the expected fraction of neighbors in state \red at round $t+1$ conditioned on the current configuration of states, namely
    \[
        \Ex{\rfrac{u}{t+1} \condconf}
        = \frac{1}{\deg{v}} \sum_{v \in \neigh{u}} \Prob{\card{\bR{v}{t}}=1 \condconf}
        = \frac{1}{\deg{v}} \sum_{v \in \neigh{u}} (1-p)\rfrac{v}{t}
        \leq (1-p) \rfracmax{t}\,.
    \]
    Note that, in the previous equation, the conditioning exactly determines the fraction of neighbors in state \red of every node at round $t$; thus we can compute the probability of every node to be in state \red in the next round. 
    Indeed, taking into account the effect of the bias, we have 
    \[
        \Prob{\card{\bR{v}{t}}=1 \condconf} = (1-p)\rfrac{v}{t}\,,
    \]
    where
    $\rfrac{v}{t}$ is deterministic due to the conditioning.
    
    As done in the proof of \cref{prop:supercritical} for the fast convergence regime, with an application of a multiplicative form of Chernoff's Bound (\cref{thm:chernoff extension}) we get that
    \[
        \Prob{\rfrac{u}{t+1} > (1-p^2)\rfracmax{t} \condconf}
        = n^{-\omega(1)}
    \]
    and thus a union bound over all the agents allows us to claim that 
    \[
        \Prob{\rfracmax{t+1} \leq (1-p^2)\rfracmax{t}\condconf}
        = 1 - n^{-\omega(1)}\,.
    \]
    Therefore, since any initial configuration $\conf{0}$ is such that $\rfracmax{0}\leq 1$, at round $t$ we have that
    $\rfrac{u}{t} \leq (1-p^2)^t$. 
    An application of \cref{coro:prop1} with $q_t=0$ implies the thesis. 
\end{proof}


\subsection{Deterministic Majority Dynamics}\label{ssec:det_maj}
As mentioned earlier, one might expect that as $k$ grows \majority{k}{p}{\blue} and \modmajority{k}{p}{\blue} would behave similarly to \determ{} \maj{} (where instead of sampling $k$ neighbors, nodes consider the entire neighborhood to determine their next state) in the biased communication models, which we call \detmaj{p}{\blue} and \moddetmaj{p}{\blue}.
We make this link rigorous in \cref{prop:det-majority,prop:kmaj-klarge}. 
In particular, we show that, if the graph satisfies the \emph{density assumption} $\min_{u}\deg{u}=\omega(\log n)$, both  \detmaj{p}{\blue} and \moddetmaj{p}{\blue} have a sharp phase transition at $p^\star=\frac{1}{2}$, which is in fact the limit as $k\rightarrow\infty$ of the critical value $\pk$ (as proved in \cref{claim:05}).
For the sake of readability, we restrict the analysis of this section to the case in which every vertex is initially in state \red.

Formally, the \determ{} \maj{} dynamics in the two biased communication models is defined as follows. 
\begin{definition}[\detmaj{p}{\blue} dynamics]\label{def:det-majority}
Let $p \in [0,1]$. Starting from an initial configuration $\conf{0}$, at each round $t$ every node $u \in V$ updates its state as
\[
    \state{u}{t} = \left\{\begin{array}{ll}
        \red
            & \text{if } |\bR{u}{t-1}| > |\bB{u}{t-1}|\,,
        \\
        \red \text{ or } \blue \text{ with probability } 1/2
            & \text{if } |\bR{u}{t-1}| = |\bB{u}{t-1}|\,,
        \\
        \blue
            & \text{if } |\bR{u}{t-1}| < |\bB{u}{t-1}|\,,
    \end{array}
    \right.
\]
where $\bR{u}{t}$ and $\bB{u}{t}$ are defined in \cref{barsets} with $S_u^{(t)} = \neigh{u}$ for every $t$.
\end{definition}

As done in \cref{sec:model}, for every $u \in V$ and every $t \in \mathbb{N}$, denote by $M_u^{(t)}$ the Bernoulli random variable of parameter $p$ which models the bias toward node $u$ in round $t$, i.e., $M_u^{(t)}=1$ with probability $p$ and $M_u^{(t)}=0$ otherwise.
\begin{definition}[\moddetmaj{p}{\blue} dynamics]\label{def:moddet-majority}
Let $p \in [0,1]$. Starting from an initial configuration $\conf{0}$, at each round $t$ every node $u \in V$ updates its state as
\[
    \state{u}{t} = \left\{\begin{array}{ll}
        \red & \text{if } M_u^{(t)}=0 \text{ and } |\R{t-1} \cap \neigh{u}| > |\B{t-1} \cap \neigh{u}|\,,
        \\
        \red \text{ or } \blue \text{ with probability $1/2$} 
            & \text{if } M_u^{(t)}=0 \text{ and } |\R{t-1} \cap \neigh{u}| = |\B{t-1} \cap \neigh{u}|\,,
        \\
        \blue & \text{if } M_u^{(t)}=1 \,\,\text{ or }\,\, |\R{t-1} \cap \neigh{u}| < |\B{t-1} \cap \neigh{u}|\,.
    \end{array}
    \right.
\]
\end{definition}

\begin{proposition}
\label{prop:det-majority}
Consider a sequence of graphs $(G_n)_{n\in\mathbb{N}}$ such that  $\min_{u\in V}\deg{u}=\omega(\log n)$.  
Fix $p\in[0,1]$. For any fixed $n$, consider the \detmaj{p}{\blue} dynamics.
\begin{itemize}
    \item \emph{Fast disruption:} if $p>\tfrac{1}{2}$, then
    \[
        \Prob{|B^{(1)}|=n}=1-o(1)\,.
    \]
    Hence, as a corollary,
    \[\Prob{\tau=1 }=1-o(1)\,. \] 
    \item \emph{Slow disruption:} if $p<\tfrac{1}{2}$, then for all $K>0$
    \[
        \Prob{\forall t\le n^K,\,\,
        |R^{(t)}|=n}=1-o(1)\,.
    \]
    Hence, as a corollary, for all $K>0$
    \[\Prob{\tau>n^K }=1-o(1)\,. \] 
\end{itemize}
\end{proposition}
\begin{proof}
Assume $p=\tfrac{1}{2}+c$, for some $c>0$. Notice that 
\(|B^{(1)}|\sim\sum_{v\in V}X_v\,, \)
where $\{X_v\}_{v\in V}$ are independent random variables with $X_v\sim\bin{1}{r_v}$ and
\[
    r_v:=\Prob{\text{Bin}(\deg{v},p)>\frac{\deg{v}}{2}=\deg{v}(p-c)}\ge1-e^{-\Theta(\deg{v})}\,.
\]
Hence, uniformly in $v\in V$,
\[
    r_v= 1-e^{-\omega(\log n)}\,.
\]
Therefore, by a union bound 
\[
    \Prob{|B^{(1)}|=n}\ge 1-ne^{-\omega(\log n)} =1-o(1)\,.
\]

Similarly, assume that $p=\tfrac{1}{2}-c$, for some $c>0$. Notice that
\(|R^{(1)}|\sim\sum_{v\in V}W_v\,, \)
where $W_v\sim\bin{1}{\ell_v}$ and
\[
    \ell_v=\Prob{\text{Bin}(\deg{v},p)\le \frac{\deg{v}}{2}=\deg{v}(p+c)}\ge1- e^{-\Theta(\deg{v})}\,.
\]
Hence, uniformly in $v\in V$,
\[
    \ell _v=1-e^{-\omega(\log n)}\,.
\]
Therefore, by a union bound 
\[
    z:=\Prob{|R^{(1)}|=n}\ge 1- n e^{-\omega(\log n)}\,.
\]
Call $T$ the first $t\ge 0$ such that $B^{(t)}\neq\emptyset$, in other words,
\[
    \Prob{\forall t\le n^K\,, \,|R^{(t)}|=n}=\Prob{T>n^K}\,.
\]
Notice that $T$ has a geometric distribution of parameter $\Prob{|R^{(1)}|\neq n}$. Hence
\[
    \Prob{\tau>t }\ge\Prob{T>t }=\Prob{|R^{(1)}|=n}^{t}= z^t\,.
\]
For every choice of $K>0$ independent of $n$, we can bound
\[
    z^t\ge (1-n^{-2K})^t\sim e^{-t/n^{2K}}\,,
\]
from which the claim follows.
\end{proof}

Note that the result for the \detmaj{p}{\blue} is much stronger than the usual \emph{disruption analysis} we used throughout the paper. In fact, if $p\geq\frac{1}{2}+c$ for some constant $c>0$, the \blue-\emph{consensus} is reached in a single round with high probability. On the other hand, if $p<\frac{1}{2}$, the initial \red-consensus lasts for any polynomial number of rounds. 
We now show that \moddetmaj{p}{\blue} exhibits a weaker behavior, in the sense that a fast disruption is reached in 2 rounds (instead of only 1), and in the slow disruption regime there exists a minority of nodes in state \blue (instead, in the \detmaj{p}{\blue}, there are no nodes in state \blue).
\begin{proposition}
\label{prop:moddet-majority}
Fix $p\in[0,1]$ and consider a sequence of graphs $(G_n)_{n\in\mathbb{N}}$ such that  $\min_{u\in V}\deg{u}=\omega(\log n)$. 
For any fixed $n$, consider the \moddetmaj{p}{\blue} dynamics. 
\begin{itemize}
    \item \emph{Fast disruption:} if $p>\tfrac{1}{2}$, then
    \[
        \Prob{|B^{(2)}|=n}=1-o(1)\,.
    \]
    Hence, as a corollary,
    \[
        \Prob{\tau=1}=1-o(1)\,.
    \]
    \item \emph{Slow disruption:} if $p<\tfrac{1}{2}$, then, called $c=1/2-p>0$, for all $K>0$
    \[
        \Prob{\forall t\le n^K,\: \frac{|R^{(t)}|}{n}\ge \frac{1+c}{2} }=1-o(1)\,.
    \]
    Hence, as a corollary, for all $K>0$
    \[
        \Prob{\tau>n^K}=1-o(1)\,.
    \]
\end{itemize}
\end{proposition}

\begin{proof}
The proof is in the same spirit of the one of \cref{prop:det-majority}, so we just highlight the main differences and leave the details to the interested reader. Fix $c>0$. At step zero all vertices are \red, hence, the probability that each of them is \blue at next step is simply given by $p$, and those events are independent. If $p=1/2+c$, after a single step w.h.p.\ all vertices have a fraction of \blue neighbors which is strictly larger than $1/2$, thanks to the density assumption. Therefore, conditionally on the latter event, with probability $1$ all vertices will be in state $\blue$ at step 2.

Concerning the scenario in which $p=1/2-c$, at step 1 w.h.p. all vertices have a fraction of \red vertices which is larger than $(1+c)/2$. Therefore, conditionally on the latter event, each vertex will be in state \blue at the next step, independently, with probability $p$. It is then possible to iterate the same argument for every polynomial number of rounds as in the proof of \cref{prop:det-majority}.
\end{proof}

In \cref{prop:kmaj-klarge} we show that the \majority{k}{p}{\blue}, for all sufficiently large values of $k$, exhibits a behavior similar to that in \cref{prop:det-majority}.
For the sake of clarity, in the next proposition we will use the notation $\mathbf{P}_k$ to denote the law of \majority{k}{p}{\blue} for a given value of $k$. We start with the analysis of the  \majority{k}{p}{\blue} dynamics, i.e., \cref{prop:kmaj-klarge}.
The corresponding results for the \modmajority{k}{p}{\blue}, i.e., the behavior shown in \cref{prop:modkmaj-klarge} is similar to that shown in \cref{prop:moddet-majority} and can be obtained in a similar way.

\begin{proposition}
\label{prop:kmaj-klarge}
Fix $p\in[0,1]$ and consider a sequence of graphs $(G_n)_{n\in\mathbb{N}}$ such that $\min_{u\in V}\deg{u}=\omega(\log n)$.
For any fixed $n$, consider the \majority{k}{p}{\blue} dynamics. For all $\gamma>0$ there exists $H=H(p,\gamma)$ such that:
\begin{itemize}
    \item \emph{Fast disruption:} if $p>\frac{1}{2}$, for all $k>H$
    \begin{equation}\label{kmj1}
        \Probb{k}{\frac{|B^{(1)}|}{n}\ge 1-\gamma}=1-o(1)\,.
    \end{equation}
    \item \emph{Slow disruption:} if $p<\frac{1}{2}$  for all $k>H$ and $K>0$
    \begin{equation}\label{kmj2}
        \Probb{k}{
        \forall t\in[0,n^K]\,,\,\,
        \frac{|R^{(t)}|}{n}\ge 1-\gamma}=1-o(1)\,.
    \end{equation}
\end{itemize}
\end{proposition}
\begin{proof}
    We first prove \cref{kmj1}. Fix $c>0$ and $p=\frac{1}{2}+c$. Assume that at round 0 all the vertices are in the state $\red$.
    Then 
    \[
        \Probb{k}{|\R{1}|>\gamma n }=\Prob{\bin{n}{r_k}>\gamma n}\,,
    \]
    where 
    \begin{align*}
         r_k&:=\Prob{\bin{k}{1-p}\ge \frac{k+1}{2}}=\Prob{\bin{k}{1-p} \ge k\left(\frac{1}{2}-c\right)\left(\frac{1}{1-2c}+\frac{1}{k(1-2c)}\right)}
         \leq e^{-M(c)k}\,,
    \end{align*}
    for some constant $M(c)>0$.
  Hence 
  \begin{equation*}
    \Prob{\text{Bin}(n,r_k)>\gamma n}=\Prob{\text{Bin}(n,r_k)>r_k n(1+\lambda)}\,,
    \end{equation*}
    where $\lambda=(\gamma/r_k)-1$. Since $r_k<e^{-M(c)k}$, there exists $\tilde K=\tilde K(c,\gamma)$ such that for $k>\tilde K$ we have $\gamma>r_k$, hence $\lambda>0$. Hence applying the Chernoff's bound (\cref{thm:chernoff extension}) we get 
   \[
    \Prob{\text{Bin}(n,r_k)>r_k n(1+\lambda)}\leq e^{-\Theta(n)}\,,
   \]
   uniformly in $k>\tilde K$, hence \cref{kmj1} holds.

    Let us now prove \cref{kmj2}. Fix $c>0$, $p=\frac{1}{2}-c$ and $\gamma>0$.
    By \cref{claim:05,prop:subcritical-general,coro:prop1} with $p=\frac{1}{2}-c$, there exists $\bar K(c)>0$ such that for any fixed $K>0$ and for all $k>\bar K(c)$ it holds
    \begin{equation}\label{kmj7}
    \Probb{k}{\forall t<n^K,
    \quad\frac{|R^{(t)}|}{n}\geq\varphi_{k,\frac{1}{2}-c}^+-\frac{\gamma}{2}}=1-o(1)\,.
    \end{equation}
    Note that, since $\varphi_{k,\frac{1}{2}-c}^+\rightarrow 1$ when $k\rightarrow +\infty$, we have that for any $\gamma>0$ there exists $\hat K(c,\gamma)$ such that for all $k>\hat K(c,\gamma)$ it holds 
    \begin{equation}\label{kmj8}
    1-\varphi_{k,\frac{1}{2}-c}^+<\frac{\gamma}{2}\,.
    \end{equation}
    So by \cref{kmj7,kmj8} we get that, fixed $p=\frac{1}{2}-c$ and any $\gamma\in(0,1)$, for all $k:=\max\{\bar K(c),\hat K(c,\gamma)\}$ we have that
    \cref{kmj2} holds. 
    
    The thesis follows by choosing $H(p,\gamma):=\max\{\bar K(1/2-p),\hat K(1/2-p,\gamma),\tilde K(1/2-p,\gamma) \}$.
\end{proof}

We point out that the same behavior of \cref{prop:det-majority} can be proved for \majority{k}{p}{\blue}, by just letting $k$ grow with $n$ as $k=\omega(\log n)$.

We now state the analogue of \cref{prop:kmaj-klarge} for the \modmajority{k}{p}{\blue} dynamics.

\begin{proposition}
\label{prop:modkmaj-klarge}
Fix $p\in[0,1]$ and consider a sequence of graphs $(G_n)_{n\in\mathbb{N}}$ such that $\min_{u\in V}\deg{u}=\omega(\log n)$. 
For any fixed $n$, consider the \modmajority{k}{p}{\blue} dynamics. For all constant $\gamma>0$ there exists $H=H(p,\gamma)$ such that:
\begin{itemize}
    \item \emph{Fast disruption:} if $p>\frac{1}{2}$, for all $k>H$
    \[
        \Prob{\tau=1}=1-o(1)\,.
    \]
    Moreover, for all $k>H$
    \[
        \Probb{k}{\frac{|B^{(2)}|}{n}\ge 1-\gamma}=1-o(1)\,.
    \]
    \item \emph{Slow disruption:} if $p<\frac{1}{2}$, for all $k>H$ and $K>0$
    \[
        \Probb{k}{\tau> n^K}=1-o(1)\,.
    \]
\end{itemize}
\end{proposition}

\begin{proof}
   Also in this case the proof is similar as that of \cref{prop:kmaj-klarge} with the same differences highlighted in the proof of \cref{prop:moddet-majority}.
\end{proof}

\appendix
\section*{Appendix}
\section{Properties of 
\texorpdfstring{$\Fpk{p}{k}(x)$}{Fpk(x)}}\label{sec:apx fpk}

Recall from \cref{def:fpk} that 
\(
    \Fpk{p}{k}(x) \dfn \Prob{\bin{k}{(1-p)x}\geq \frac{k+1}{2}}
\)
for every odd $k$.
In this section we study the fixed points of function $\Fpk{p}{k}(x)$, i.e., we want to solve the equation 
\(\Fpk{p}{k}(x) = x\) for every $p \in (0,1)$.
We first consider the simplest case $k=3$.

\begin{lemma}\label{lem:fixed_3}
	Consider the equation $\Fpk{p}{3}(x) = x$.
	It holds that:
	\begin{itemize}
		\item if $p<\frac{1}{9}$ then there exist $\varphi^-,\varphi^+ \in [\frac{1}{2(1-p)},1]$ 
		such that $\Fpk{p}{3}(x) = x$ has solution $0$, $\varphi^-$, and $\varphi^+$, with $\varphi^-<\varphi^+$;
		\item if $p=\frac{1}{9}$ then there exists $\varphi\in [\frac{1}{2(1-p)},1]$ 
		such that $\Fpk{p}{3}(x) = x$ has solution $0$ and $\varphi$;
		\item if $p>\frac{1}{9}$ then $\Fpk{p}{3}(x) = x$ has $0$ as unique solution.
	\end{itemize}
	
\end{lemma}
\begin{proof}
	With $k=3$ we have that 
	\begin{align*}
		 \Fpk{p}{3}(x) &= (1-p)^3x^3 - 3(1-p)^2x^2[1-(1-p)x]\\
		            &= -2(1-p)^3x^3 + 3(1-p)^2x^2.
	\end{align*}
	With some algebraic manipulations on the equation $\Fpk{p}{3}(x) = x$,
	we get that
	\begin{equation}\label{eq:fixed_3}
		\Fpk{p}{3}(x) = x
		\quad\iff\quad
		x[2(1-p)^3x^2 - 3(1-p)^2x + 1]=0.
	\end{equation}
	Note that \cref{eq:fixed_3} has solutions different from $0$ if and only if $2(1-p)^3x^2 - 3(1-p)^2x + 1=0$ has at least one solution, i.e., if the discriminant	$\Delta = (1-p)^3(1-9p) \geq 0$.
	If $p = \frac{1}{9}$ then $\Delta = 0$ and thus the unique solution of \cref{eq:fixed_3} different from 0 is
	\[
	    \varphi =\frac{27}{32}>\frac{9}{16}= \frac{1}{2(1-p)}\geq\frac{1}{2}.
	\]
	If $p<\frac{1}{9}$, instead, \cref{eq:fixed_3} has two solutions different from $0$ which are
	\[
	    \varphi^{\pm} = \frac{3(1-p)^2\pm \sqrt{(1-p)^3(1-9p)}}{4(1-p)^3} \geq \frac{1}{2(1-p)}\geq\frac{1}{2}.
	    \qedhere
	\]
\end{proof}

\bigskip
We now look at the function $\Fpk{p}{k}$ for generic values of $k$. 
Before stating and proving the version of \cref{lem:fixed_3} for generic $k$, we show a set of properties of the function $\Fpk{p}{k}$ in the following \cref{notaapp,claim:dec_inc_F_k,claim:der1_F,claim:der2_F,claim:convex}
\begin{note}\label{notaapp}
    Recall that $\Fpk{p}{k}(x) \dfn \Prob{\bin{k}{(1-p)x}\geq \frac{k+1}{2}}$. 
    Note that:
	\begin{enumerate}
		\item\label{prop1} $\Fpk{p}{k}(0)=0$ for every $p\in (0,1)$ and every odd $k$.
		\item\label{prop2} $\Fpk{p}{k}(1)< 1$ for every $p\in (0,1)$ and every odd $k$.
		\item\label{prop3} $\Fpk{p}{k}(x)$ is a continuous function in the variable $x$ for every $p\in (0,1)$ and every $k$.
		\item\label{prop4} $\Fpk{p}{k}(x)$ is increasing in the variable $x$ for every $p\in (0,1)$ and every odd $k$.
	     In fact, if we take $x<y$, we have that
	    \[
		    \Prob{\bin{k}{(1-p)x}\geq \frac{k+1}{2}} < \Prob{\bin{k}{(1-p)y}\geq \frac{k+1}{2}}.
		\]
		\item\label{prop5} $\Fpk{p}{k}(x)$ is continuous and decreasing in $p$ for any $x\in[0,1]$ and for every odd $k$.
		 \item\label{prop7} If $k$ is odd and $p<\frac{1}{2}$, we have $\frac{1}{2(1-p)}\in[0,1]$ and
       $\Fpk{p}{k}\left(\frac{1}{2(1-p)}\right)= \frac{1}{2}$. In fact,
       \begin{align*}
       \Fpk{p}{k}\left(\frac{1}{2(1-p)}\right) & = 
       \Prob{\bin{k}{\frac{1}{2}} \geq \frac{k+1}{2}} 
       = \sum_{i=(k+1)/2}^k \binom{k}{i}\left(\frac{1}{2}\right)^i \left(\frac{1}{2}\right)^{k-i} \\
       & = \left(\frac{1}{2}\right)^k \sum_{i=(k+1)/2}^k \binom{k}{i}
       \stackrel{(a)}{=} \left(\frac{1}{2}\right)^k  \frac{2^k}{2} 
       = \frac{1}{2}
       \end{align*}
       where in $(a)$ we use that $\binom{k}{i} = \binom{k}{k-i}$ and $\sum_{i=0}^k \binom{k}{i} = 2^k$.
		\item\label{prop6} $\Fpk{p}{k}\left(\frac{1}{2}\right)
		\leq \frac{1}{2}$ 
        for every $p\in [0,1]$ and every odd $k$. In fact, by \cref{prop4}
        \begin{align*}
        \Fpk{p}{k}\left(\frac{1}{2}\right)\le \Fpk{p}{k}\left(\frac{1}{2(1-p)}\right)=\frac{1}{2}.
        \end{align*}
	\end{enumerate}
\end{note}

\begin{claim}\label{claim:dec_inc_F_k}
	 The function $\Fpk{p}{k}(x)$ is non-decreasing in $k$ for every $x\in\left[\min\left\{ \frac{1}{2(1-p)},1\right\},1\right]$ and  $p\in[0,1]$, 
	while it is non-increasing in $k$ for every $x\in\left[0,\min\left\{ \frac{1}{2(1-p)},1\right\}\right]$ and $p\in[0,1]$.
\end{claim}
\begin{proof}
	Let $X$ and $Y$ be two random variables 
	with laws $\bin{2h+1}{q}$ and $\bin{2h+3}{q}$, respectively.
	In the following we prove that
	\begin{align}\label{eq:bin_geq}
		\Prob{Y\geq h+2} \geq \Prob{X\geq h+1}
		&\text{ if } q \geq \frac{1}{2},
		\\\label{eq:bin_leq}
		\Prob{Y\geq h+2} \leq \Prob{X\geq h+1} 
		&\text{ if } q \leq \frac{1}{2}.
	\end{align}
	Note that
	\begin{align*}
		\Prob{Y\geq h+2} &=
		\Prob{X \geq h+2} +
		\Prob{X = h+1}(1-(1-q)^2) +
		\Prob{X = h}q^2\,.
	\end{align*}
	Since 
	\[
	    \Prob{X=h}=\binom{2h+1}{h}q^h(1-q)^{h+1}=\binom{2h+1}{h+1}q^h(1-q)^{h+1}=\frac{1-q}{q}\Prob{X=h+1}\,,
	\]
	\cref{eq:split_bin} can be rewritten as
	\begin{align}
		\Prob{Y\geq h+2} &=
		\Prob{X \geq h+2} +
		((1-(1-q)^2) +q(1-q))\Prob{X = h+1}=
		\notag\\
		\label{eq:split_bin}
		&=\Prob{X \geq h+2} +
		q(3-2q)\Prob{X = h+1}\,.
    \end{align}
	Now we compute 	$\Prob{Y\geq h+2}- \Prob{X\geq h+1}$ by using \cref{eq:split_bin}. We get that 
	\begin{align*}
	&\Prob{Y\geq h+2} - \Prob{X \geq h+1} =
	\Prob{X \geq h+2} + q(3-2q) \Prob{X = h+1} - \Prob{X \geq h+1} 
	\\
	&=\Prob{X \geq h+2} + q(3-2q)  \Prob{X = h+1} - \Prob{X = h+1} - \Prob{X \geq h+2}
	\\
	&=(3q-2q^2-1)\Prob{X = h+1}
	\end{align*}
	whose sign depends only on the factor $(3q-2q^2-1)$. Since
	$3q-2q^2-1>0$ if and only if  $q\in\left[\frac{1}{2},1\right]$, 
	we can conclude that
	$\Prob{Y\geq h+2} \geq \Prob{X \geq h+1}$ for every $q\geq \frac{1}{2}$ and
	$\Prob{Y\geq h+2} \leq \Prob{X \geq h+1}$ for every $q\leq \frac{1}{2}$.
	
	Let $k=2h+1$. From \cref{eq:bin_geq,eq:bin_leq} we have that
	$\Fpk{p}{k}(x) $ is increasing in $k$ for every $x\geq \frac{1}{2(1-p)}$ while it is decreasing in $k$ for every $x\leq \frac{1}{2(1-p)}$.
\end{proof}

\begin{claim}\label{claim:der1_F}
It holds that
\[
	\frac{d}{dx}\Fpk{p}{k}(x) 
	= k(1-p)\Prob{\bin{k-1}{(1-p)x}= \frac{k-1}{2}}.
\]
\end{claim} 
\begin{proof}
	We start by showing that 
	\[
		\frac{d}{du}\Prob{\text{Bin}(k,u)\leq \frac{k-1}{2}}
		=-k\Prob{\text{Bin}(k-1,u)=\frac{k-1}{2}}.
	\]
	Observe that
	\begin{align}\label{der1}
	\frac{d}{du}\Prob{\text{Bin}(k,u)\leq \frac{k-1}{2}}
	&=\sum_{i=0}^{\frac{k-1}{2}}\binom{k}{i}\frac{d}{du}[u^i(1-u)^{k-i}]
	\notag\\
	&=\frac{d}{du}[(1-u)^{k}]+\sum_{i=1}^{\frac{k-1}{2}}\binom{k}{i}\frac{d}{du}[u^i(1-u)^{k-i}].
	\end{align}
	We consider separately the term $\frac{d}{du}[u^i(1-u)^{k-i}]$ for $i\geq 1$. 
	We have
	\begin{align}\label{der2}
	\frac{d}{du}[u^i(1-u)^{k-i}]=iu^{i-1}(1-u)^{k-i}-(k-i)u^i(1-u)^{k-1-i}\,.
	\end{align}
	Note that 
	\begin{align}\label{der3}
	\sum_{i=1}^{\frac{k-1}{2}}\binom{k}{i}iu^{i-1}(1-u)^{k-i}
	&=\sum_{i=1}^{\frac{k-1}{2}}\frac{k!}{(i-1)!(k-i)!}u^{i-1}(1-u)^{k-i}
	\notag\\
	&\stackrel{(j=i-1)}{=} \sum_{j=0}^{\frac{k-3}{2}}\frac{k!}{j!(k-1-j)!}u^{j}(1-u)^{k-1-j}
	\end{align}
	and		
	\begin{align}\label{der4}
	\sum_{i=1}^{\frac{k-1}{2}}\binom{k}{i}(k-i)u^i(1-u)^{k-1-i}=
	\sum_{i=1}^{\frac{k-1}{2}}\frac{k!}{i!(k-1-i)!}u^i(1-u)^{k-1-i}\,.
	\end{align}
	So by \cref{der1,der2,der3,der4} we get	
	\begin{align}\label{der5}
	&\frac{d}{du}\Prob{\text{Bin}(k,u)\leq \frac{k-1}{2}}=
	\frac{d}{du}[(1-u)^{k}]+\sum_{i=1}^{\frac{k-1}{2}}\binom{k}{i}
	\frac{d}{du}[u^i(1-u)^{k-i}]
	\notag\\
	&=-k(1-u)^{k-1}+
	\sum_{j=0}^{\frac{k-3}{2}}\frac{k!}{j!(k-1-j)!}u^{j}(1-u)^{k-1-j}-
	\sum_{i=1}^{\frac{k-1}{2}}\frac{k!}{i!(k-1-i)!}u^i(1-u)^{k-1-i}
	\notag\\
	&=-k(1-u)^{k-1}+k(1-u)^{k-1}-
	\frac{k!}{\left(\frac{k-1}{2}\right)!\left(\frac{k-1}{2}\right)!}u^{\frac{k-1}{2}}(1-u)^{\frac{k-1}{2}}
	\notag\\
	&=-k\binom{k-1}{\frac{k-1}{2}}u^{\frac{k-1}{2}}(1-u)^{\frac{k-1}{2}}
	=-k\Prob{\text{Bin}(k-1,u)=\frac{k-1}{2}}\,.
	\end{align}
	We want to use \cref{der5} for computing the derivative of $\Fpk{p}{k}(x)$.
	Note that  
	\begin{align}\label{der_eq}
		\frac{d}{dx}\Fpk{p}{k}(x) &= 
		\frac{d}{dx}\left( 1 - \Prob{\bin{k}{(1-p)x}<\frac{k+1}{2}}\right)
		\notag\\ 
		&= - \frac{d}{dx}\left( \Prob{\bin{k}{(1-p)x}\leq\frac{k-1}{2}}\right).
	\end{align}
	If we call $u(x)=(1-p)x$ we have that 
	\(
	    \frac{d}{dx}\Fpk{p}{k}(x) = \frac{d}{du}\left(\Fpk{p}{k}(u)\right)
	    \frac{d}{dx}\left(u(x)\right).
	\)
	Thus, by \cref{der_eq} we can conclude that
	\begin{align*}
		\frac{d}{dx}\Fpk{p}{k}(x) &= 
		\frac{d}{du}\left(\Fpk{p}{k}(u)\right)\frac{d}{dx}\left(u(x)\right) =
		-\frac{d}{du}\left( \Prob{\bin{k}{u}\leq\frac{k-1}{2}}\right) (1-p)
		\\
		&= k(1-p)\Prob{\bin{k-1}{(1-p)x}= \frac{k-1}{2}}.
		\qedhere
	\end{align*}	
\end{proof}

\begin{claim}\label{claim:der2_F}
	It holds that
	\[
	    \frac{d^2}{dx^2}\Fpk{p}{k}(x) 
	    = k(k-1)(1-p)^3 \binom{k-2}{\frac{k-1}{2}} (x-(1-p)x^2)^{\frac{k-3}{2}}(1-2(1-p)x).
	\]
\end{claim} 
\begin{proof}
	We start by showing that 
	\begin{equation}\label{der1_bin}
		\frac{d}{du}\left(\Prob{\bin{k-1}{u} = \frac{k-1}{2}}\right) =
		(k-1) \binom{k-2}{\frac{k-1}{2}}(u - u^2)^{\frac{k-3}{2}}(1-2u).
	\end{equation}
	Note that
	\begin{align*}
	\frac{d}{du}\left(\Prob{\text{Bin}(k-1,u)=\frac{k-1}{2}}\right)
	&=\binom{k-1}{\frac{k-1}{2}}\frac{d}{du}\left[u^{\frac{k-1}{2}}(1-u)^{\frac{k-1}{2}}\right]
	\\
	&=\binom{k-1}{\frac{k-1}{2}}\frac{d}{du}(u-u^2)^{\frac{k-1}{2}}
	\\
	&=\binom{k-1}{\frac{k-1}{2}}\frac{k-1}{2}(u-u^2)^{\frac{k-3}{2}}(1-2u)
	\\
	&=(k-1)\binom{k-2}{\frac{k-1}{2}}(u-u^2)^{\frac{k-3}{2}}(1-2u).
	\end{align*}
	If we call $u(x)=(1-p)x$ we can compute the second derivative of $\Fpk{p}{k}(x)$ by using \cref{der1_bin}. 
	Indeed,
	\begin{align*}
		\frac{d^2}{dx^2}\Fpk{p}{k}(x) &=
		k(1-p)\frac{d}{dx}\left(\Prob{\bin{k-1}{(1-p)x}= \frac{k-1}{2}}\right)
		\\
		&= k(1-p)\frac{d}{du}\left(\Prob{\bin{k-1}{u}=\frac{k-1}{2}}\right)
		\frac{d}{dx}(u(x))
		\\
		&= k(k-1)(1-p)^2 \binom{k-2}{\frac{k-1}{2}} ((1-p)x-(1-p)^2x^2)^{\frac{k-3}{2}}(1-2(1-p)x)
		\\
		&= k(k-1)(1-p)^{\frac{k+1}{2}} \binom{k-2}{\frac{k-1}{2}} (x-(1-p)x^2)^{\frac{k-3}{2}}(1-2(1-p)x).
		\qedhere
	\end{align*}
\end{proof}

\begin{claim}\label{claim:convex}
    For $p\in(0,1)$ and for any odd $k$, the map $x\mapsto \Fpk{p}{k}(x)$ is a convex function for $x \in \left[0, \min\left\{\frac{1}{2(1-p)}, 1\right\} \right)$.
\end{claim}
\begin{proof} By direct computation:
	\[
		[x(1-(1-p)x)]^{\frac{k-3}{2}}(1-2(1-p)x) > 0 
		\,\iff\,
		1-2(1-p)x > 0 
		\,\iff\,
		x < \frac{1}{2(1-p)}.
		\qedhere
	\]
\end{proof}

\bigskip
We are now ready to state and prove the generalized version of \cref{lem:fixed_3}.
\fixedpoints*
\begin{figure}[ht]
\begin{center}\hfill%
    \begin{minipage}[t]{0.49\textwidth}
      \includegraphics[width=\linewidth]{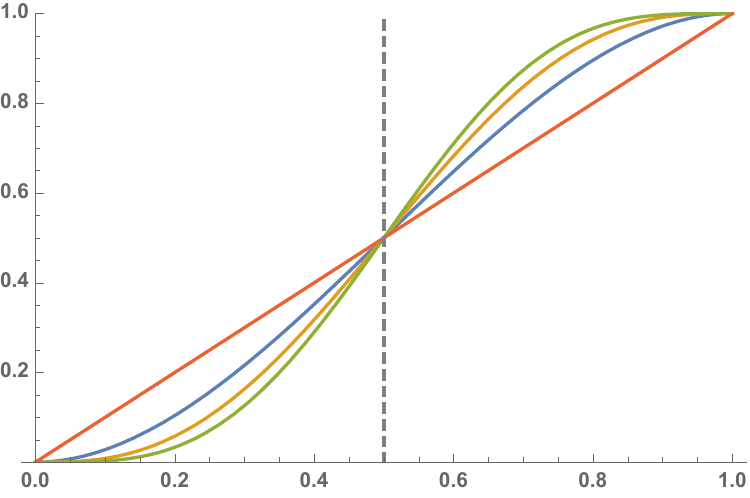}
    \end{minipage}\hfill%
    \begin{minipage}[t]{0.49\textwidth}
      \includegraphics[width=\linewidth]{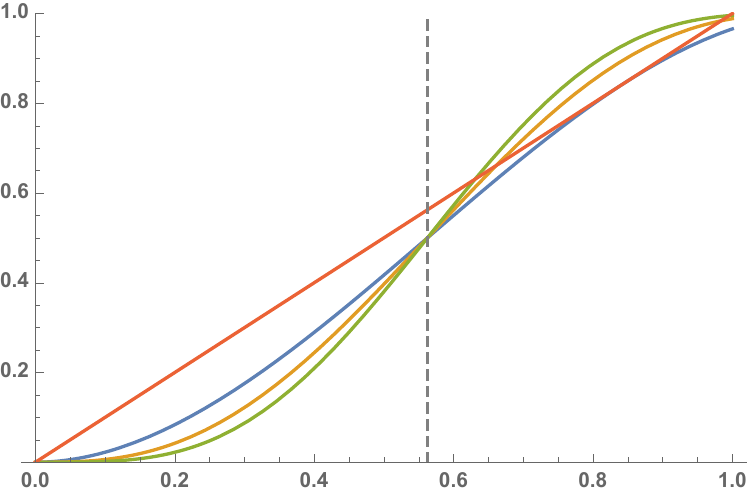}
    \end{minipage}\hfill%
    
    \hfill%
    \begin{minipage}[t]{0.49\textwidth}
      \includegraphics[width=\linewidth]{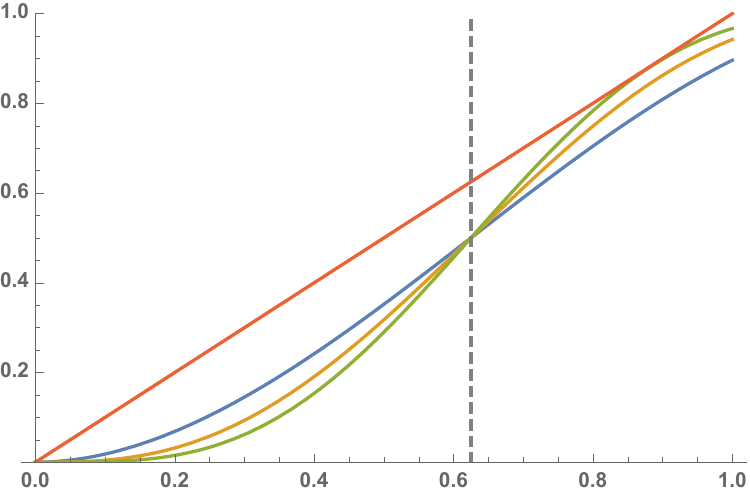}
    \end{minipage}\hfill%
    \begin{minipage}[t]{0.49\textwidth}
      \includegraphics[width=\linewidth]{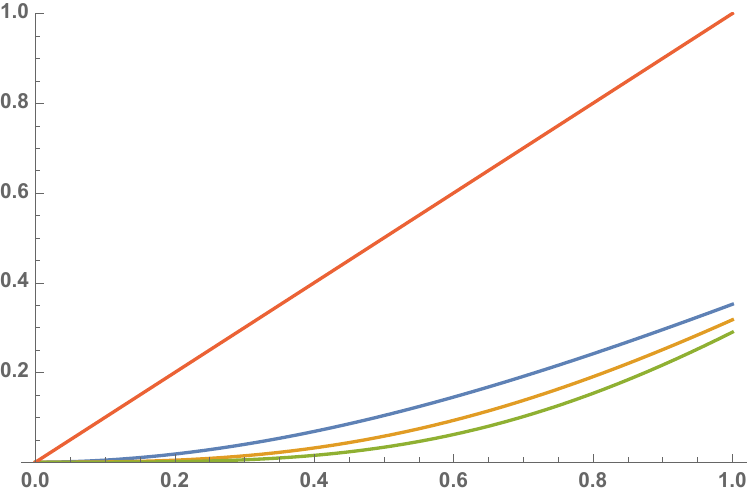}
    \end{minipage}\hfill%
\end{center}
\caption{Plot of the function $x\mapsto \Fpk{p}{k}(x)$. The blue, orange, and green curves represent the cases $k=3,5$ and $7$, respectively. From the top left corner in clockwise order we have the cases $p=0, \frac{1}{9}, \frac{3}{5}$ and $\frac{1}{5}$, respectively. The dashed line has equation $x=1/(2(1-p))$ and intersects the curves in their inflection point. The red line is the diagonal of the positive orthant.}
    \label{fig:my_label}
\end{figure}
\begin{proof}
Recall that, by \cref{claim:convex}, the map $x\mapsto \Fpk{p}{k}(x)$ is convex in $(0,1)$ for all $p\in(1/2,1)$ and $k$ odd. Moreover, being $\Fpk{p}{k}(0)=0$ (\cref{prop1} in \cref{notaapp}) and $\Fpk{p}{k}(1)<1$ (\cref{prop2} in \cref{notaapp}), regardless of the value of $k$ there are no solutions to the equation in the variable $x$
\[
    \Fpk{p}{k}(x)=x,
\]
as soon as $p>1/2$.
By \cref{prop7} in \cref{notaapp} we know that, for all $k$ odd and $p\in(0,1/2)$,
\[
    \Fpk{p}{k}\left(\frac{1}{2(1-p)}\right)=\frac{1}{2}< \frac{1}{2(1-p)}.
\]
Therefore, for all $k$ odd and $p\in(0,1)$ there are no solutions in $\big(0,\tfrac{1}{2(1-p)}\big]$ to the equation
\[
    \Fpk{p}{k}(x)=x,
\]
since it holds $\Fpk{p}{k}(x)<x.$

Hence, we now look for possible solutions in the interval $\left(\tfrac{1}{2(1-p)},1\right)$.  The case $k=3$ has been studied in \cref{lem:fixed_3}. Note that 
\begin{enumerate}[(i)]
    \item\label{point1} By \cref{claim:dec_inc_F_k}, the map $k\mapsto\Fpk{p}{k}(x)$ is increasing for all $p\in(0,1/2)$ and $x\in\left(\tfrac{1}{2(1-p)},1\right)$.
    \item\label{point2} By \cref{prop3} in \cref{notaapp} and by \cref{claim:convex}, the map $x\mapsto\Fpk{p}{k}(x)$ is continuous and concave in $\left(\tfrac{1}{2(1-p)},1\right)$ for all $k$ odd and $p\in(0,1/2)$.
    \item\label{point3} By \cref{prop5} in \cref{notaapp}, the map $p\mapsto\Fpk{p}{k}(x)$ is continuous and decreasing for all $k$ odd and $x\in(0,1)$.
\end{enumerate}
	
So, by \cref{point1}, for all $k>3$ there are two solutions in $(1/2,1)$ to the equation
\(
    \Fpk{\tfrac{1}{9}}{k}(x)=x.
\)

By \cref{point2} and \cref{point3}, for all $k>3$ there exists a value $p^\star_k\in\left(\tfrac{1}{9},\tfrac{1}{2} \right)$ for which the solution to the equation $\Fpk{p_k^\star}{k}(x)=x$ is unique in the interval $(\tfrac{1}{2},1)$, while for $p>p_k^\star$ the equation
\(
    \Fpk{p}{k}(x)=x,
\)
has no solutions in $(1/2,1)$.

Note that $\Fpk{p}{k}(1)<1$  (see \cref{prop2} in \cref{notaapp}) and at the beginning of the proof we have shown that  $\Fpk{p}{k}(x)=x$ has no solutions in $\left(0,\frac{1}{2(1-p)}\right]$ for all $k$ odd and $p\in(0,1)$. Then we have that the possible solutions different from 0 of $\Fpk{p}{k}(x)=x$ are contained in the interval $\left(\frac{1}{2(1-p)},1\right)\subset(1/2,1)$. 
\\Moreover, by \cref{point1}, we have that the sequence $(p_k^\star)_{k\in 2\mathbb{N}+1}$ is increasing.

To conclude the proof we have to study the sign of $\psi(x):=\Fpk{p}{k}(x)-x$. Note that, by \cref{prop4} in \cref{notaapp}, $\psi$ is a continuous in the variable $x$.
Moreover, by \cref{claim:convex}, $\Fpk{p}{k}(x)$ is strictly convex in $\left(0,\frac{1}{2(1-p)}\right)$ and hence
\begin{equation}\label{alconv}
\Fpk{p}{k}(x)< \frac{\Fpk{p}{k}\left(\frac{1}{2(1-p)}\right)}{\frac{1}{2(1-p)}}\cdot x=(1-p)x\leq x
\end{equation}
for all $x\in\left[0,\frac{1}{2(1-p)}\right]$. So we conclude that $\psi(x)<0$ for $x\in\left(0,\frac{1}{2(1-p)}\right).$

Let us define $y=\min\{h\in(0,1]\,:\,\psi(h)=0\}\,,$
where we assume that $y=1$ if $\{h\in(0,1]\,:\,\psi(h)=0\}=\emptyset$.
By \cref{alconv}, we have that $y>\frac{1}{2(1-p)}$. Since $\psi(x)<0$ for $x\in\left(0,\frac{1}{2(1-p)}\right)$ and $\psi$ is continuous, we have that $\psi(x)<0$ for all $x\in(0,y)$.
This implies that 
\begin{itemize}
    \item if $p<\pk$, then $y=\varphi_{p,k}^-$ and hence $\psi(x)<0$ for $x\in\left(0,\varphi_{p,k}^-\right)$;
    \item if $p=\pk$, then $y=\varphi_{p,k}$ and hence $\psi(x)<0$ for $x\in\left(0,\varphi_{p,k}\right)$;
    \item if $p>\pk$, then $y=1$ and hence $\psi(x)<0$ for $x\in(0,1)$. In particular, since $\Fpk{p}{k}(1)<1$, then $\psi(x)<0$ also when $x=1$.
\end{itemize}
Now let us define $r=\max\{h>0\,:\,\psi(h)=0\}\,,$ where we assume that $r=1$ if $\{h\in(0,1]\,:\,\psi(h)=0\}=\emptyset$. Since $\psi$ is continuous and $\psi(1)<0$, by definition of $r$ we have that $\psi(x)<0$ for all $x\in (r,1]$. This implies that 
\begin{itemize}
    \item if $p<\pk$, then $r=\varphi_{p,k}^+$ and hence $\psi(x)<0$ for $x\in\left(\varphi_{p,k}^+,1\right]$;
    \item if $p=\pk$, then $r=\varphi_{p,k}$ and hence $\psi(x)<0$ for $x\in\left(\varphi_{p,k},1\right]$.
\end{itemize}
To conclude the proof we have to analyze the sign of $\psi(x)$ when $p<\pk$ and $x\in\left(\varphi^-_{p,k},\varphi^+_{p,k}\right)$. Note that if $p<\pk$, then $\psi(x)=0$ if $x\in\{0,\varphi^-_{p,k},\varphi^+_{p,k}\}$. 	Moreover, since $\frac{1}{2(1-p)}<\varphi^-_{p,k}<\varphi^+_{p,k}$, then by \cref{claim:convex} we have that $\Fpk{p}{k}(x)$ is strictly concave in $(\varphi^-_{p,k},\varphi^+_{p,k})$. So we have
\[
    \Fpk{p}{k}(x)> \frac{\Fpk{p}{k}(\varphi^+_{p,k})-\Fpk{p}{k}(\varphi^-_{p,k})}{\varphi^+_{p,k}-\varphi^-_{p,k}}\cdot x=x
\]
for $x\in (\varphi^-_{p,k},\varphi^+_{p,k})$ and hence $\psi(x)>0$ in the same interval.
\end{proof}
	
\begin{lemma}\label{lem:phi-}
The map $p\in[0,\pk]\mapsto \varphi^-_{p,k}$ is increasing.
\end{lemma}
\begin{proof}
 Suppose that $p<p'\leq\pk$. Then $\varphi^-_{p,k}$, $\varphi^+_{p,k}$, $\varphi^-_{p',k}$ and $\varphi^+_{p',k}$ are well defined (in case of $p'=\pk$, we have $\varphi^-_{p',k}=\varphi^+_{p',k}$). Moreover, by \cref{prop5} in \cref{notaapp}, since $p<p'$, we have \begin{equation}\label{mappa1}
 \Fpk{p'}{k}(x)< \Fpk{p}{k}(x)\quad\forall\,x\in[0,1]\,.
 \end{equation}
 In particular 
\begin{equation}\label{phi-1}
\Fpk{p'}{k}(\varphi^-_{p,k})<\Fpk{p}{k}(\varphi^-_{p,k})=\varphi^-_{p,k}\,.
\end{equation}
Consider now the function $\psi(x)=\Fpk{p'}{k}(x)-x$. By \cref{phi-1} we have $\psi(\varphi^-_{p,k})<0$. Moreover, by \cref{lem:fixed_points_F-prel}, we know that
\[\psi(x)\begin{cases}
>0, &\text{if }x\in\left(\varphi^-_{p',k},\varphi^+_{p',k}\right)\,,
\\<0, &\text{if }x<\varphi^-_{p',k}\text{ or } x>\varphi^+_{p',k}\,,
\\=0, &\text{if }x\in\left\{\varphi^-_{p',k},\varphi^+_{p',k}\right\}\,.
\end{cases}\]
Hence, since $\psi(\varphi^-_{p,k})<0$, we have $\varphi^-_{p,k}< \varphi^-_{p',k}$ or $\varphi^-_{p,k}> \varphi^+_{p',k}$. Suppose now that $\varphi^-_{p,k}> \varphi^+_{p',k}$. Then, by definition of $\varphi^-_{p,k}$ and $\varphi^+_{p',k}$, for any $x\in\left[\varphi^-_{p',k},\varphi^+_{p',k}\right]$ we have $F_{p',k}(x)\geq x>F_{p,k}(x)$, that is in contradiction with \cref{mappa1}. So it must be $\varphi^-_{p,k}< \varphi^-_{p',k}$ and hence we get the thesis.
\end{proof}

\begin{lemma}\label{claim:L}
Assume $p=\pk-c$ for some constant $c>0$. There exists a constant $\mu_{p,k}\in(\varphi^-_{p,k},\varphi_{p,k}^+)$ such that, for all $x\in(\mu_{p,k},1]$, it holds that
\[
    \Fpk{p}{k}'(x)\in[0,1).
\]
\end{lemma}

\begin{proof}
By contradiction suppose that $\Fpk{p}{k}^{\,'}(\varphi^+_{p,k})\ge 1$. Since $\Fpk{p}{k}^{\,''}(x)<0$ for all $x\in(\tfrac{1}{2(1-p)},1)$, we have
\begin{equation}\label{eq:L}
\Fpk{p}{k}'(x)>\Fpk{p}{k}'(\varphi_{p,k}^+)\ge 1,\qquad\forall x\in\big(\tfrac{1}{2(1-p)},\varphi_{p,k}^+\big).
\end{equation}
Define
\(
    \psi(x):=\Fpk{p}{k}(x)-x.
\)
By \cref{eq:L},
\[
    \psi'(x)=\Fpk{p}{k}'(x)-1>0, 
    \qquad\forall x\in\big( \tfrac{1}{2(1-p)},\varphi_{p,k}^+\big).
\]
Hence,
\[
    0=\psi(\varphi_{p,k}^+)>\psi(x),
    \qquad\forall x\in\big( \tfrac{1}{2(1-p)},\varphi_{p,k}^+\big).
\]
In particular $\psi(x)<0$ implies
\begin{equation}\label{eq:L2}
\Fpk{p}{k}(x)<x,\qquad\forall x\in\big( \tfrac{1}{2(1-p)},\varphi_{p,k}^+\big).
\end{equation}
On the other hand, the value $\varphi_{p,k}^-$ lies in the interval $\big( \tfrac{1}{2(1-p)},\varphi_{p,k}^+\big)$, and $\Fpk{p}{k}(\varphi_{p,k}^-)=\varphi_{p,k}^-$, which is in contradiction with \cref{eq:L2}. The claim follows by continuity of $\Fpk{p}{k}$.
\end{proof}

In the following lemma we consider the limit of the critical bias $\pk$  as $k\to \infty$. 

\begin{lemma}\label{claim:05}
Consider the sequence $\{\pk\}_{k\in 2\mathbb{N}+1}$. Then we have
\[
     \lim_{k\to\infty}\pk=\frac{1}{2}.
\]
\end{lemma}
\begin{proof}
Recall that 
\[\Fpk{p}{k}(x):=\mathbf{P}\left(\text{Bin}(k,(1-p)x)\geq \frac{k+1}{2}\right)\,.\]
Observe that $\lim_{k\rightarrow\infty}\Fpk{p}{k}(\frac{1}{2(1-p)})=\frac{1}{2}$. We want to compute the limit $\lim_{k\rightarrow\infty}\Fpk{p}{k}(x)$ for any $x\in(0,1)\setminus \{\frac{1}{2(1-p)}\}$. 
Let us define $m_{p,x}$ and $s_{p,x}$, respectively the normalized expectation and variance of the random variable $\text{Bin}(k,(1-p)x)$, that is
\[m_{p,x}:=(1-p)x\,,\qquad s_{p,x}:=(1-p)x(1-(1-p)x)\,.\]
Note that 
\begin{equation}\label{cltapp}
\mathbf{P}\left(\text{Bin}(k,(1-p)x)\geq \frac{k+1}{2}\right)=\mathbf{P}\left(\frac{\text{Bin}(k,(1-p)x)-k m_{p,x}}{\sqrt{k s_{p,x}}}\geq \frac{k+1-2k m_{p,x}}{2\sqrt{k s_{p,x}}}\right)
\end{equation}
and by the Central Limit Theorem for $k\rightarrow +\infty$ the r.h.s.\ of \cref{cltapp} behaves like 
\begin{equation}\label{cltapp2}
1-\Phi\left(\frac{k+1-2k m_{p,x}}{2\sqrt{k s_{p,x}}}\right)\,,
\end{equation}
where $\Phi(x):=\mathbf{P}(Z\leq x)$ and $Z$ is a random variable with standard normal distribution. 
Since 
\begin{equation}\label{cltapp3}
\lim_{k\rightarrow+\infty}\frac{k+1-2k m_{p,x}}{2\sqrt{ks_{p,x}}}=\lim_{k\rightarrow+\infty}\frac{\sqrt{k}(1-2(1-p)x)}{2\sqrt{s_{p,x}}}\rightarrow
\begin{cases}
+\infty, &\text{if }x\in\left(0,\frac{1}{2(1-p)}\right),
\\-\infty, &\text{if }x\in\left(\frac{1}{2(1-p)},1\right)\,,
\end{cases}
\end{equation}
by \cref{cltapp,cltapp2,cltapp3} we get that for any $x\in(0,1)\setminus \{\frac{1}{2(1-p)}\}$
\[
\lim_{k\rightarrow+\infty}\Fpk{p}{k}(x)=\begin{cases}
0, &\text{if }x\in\left(0,\frac{1}{2(1-p)}\right),
\\1, &\text{if }x\in\left(\frac{1}{2(1-p)},1\right).
\end{cases}
\]
In particular for any $\eta\in(0,1)$ and $x\in(0,1)\setminus \{\frac{1}{1+\eta}\}$ we have
\begin{equation}\label{cltapp5}
\lim_{k\rightarrow+\infty}\Fpk{\frac{1-\eta}{2}}{k}(x)=\begin{cases}
0, &\text{if }x\in\left(0,\frac{1}{1+\eta}\right),
\\1, &\text{if }x\in\left(\frac{1}{1+\eta},1\right).
\end{cases}
\end{equation}
By \cref{cltapp5}, for any $\gamma,\eta\in(0,1)$ there exists a constant $\bar K(\gamma,\eta)>0$ such that
\begin{equation}\label{cltapp6}
\Fpk{\frac{1-\eta}{2}}{k}\left(\frac{1}{1+\frac{\eta}{2}}\right)>1-\gamma\,,\qquad \forall k>\bar K(\gamma,\eta)\,.
\end{equation}
Fix $\eta\in(0,1)$ and $\gamma=\frac{\frac{\eta}{2}}{1+\frac{\eta}{2}}$. Hence by \cref{cltapp6} there exists a constant $\bar K=\bar K(\gamma,\eta)>0$ such that 
\[
\Fpk{\frac{1-\eta}{2}}{k}\left(\frac{1}{1+\frac{\eta}{2}}\right)>\frac{1}{1+\frac{\eta}{2}}\,,\qquad \forall k>\bar K\,,
\]
and this implies 
\[\pk>\frac{1-\eta}{2}\,,\qquad \forall k>\bar K\,.\]
So we have obtained that for any $\eta\in(0,1)$ there exists a constant $\bar K>0$, depending only on $\eta$, such that 
$\pk>\frac{1-\eta}{2}$ for all $k>\bar K$. Moreover we know that $\{\pk\}_{k\in\mathbb{N}}$ is an increasing sequence bounded from above by $\frac{1}{2}$ (see \cref{lem:fixed_points_F-prel}). By these two facts we get the thesis.
\end{proof}

\section{Concentration inequalities}
\begin{theorem}[Chernoff Bound~\cite{dubhashi2009concentration}]\label{thm:chernoff extension}
Let $X_{1}, \ldots, X_{n}$ be independent random variables taking values in $\{0,1\}$.
Let $X = \sum_{i=1}^n X_i$ and let $\mu = \Ex{X}$. 
Let $\mu_L,\mu_U$ be such that $\mu_L \leqslant \mu \leqslant \mu_U$. 
Then, for $0 <\delta < 1$,
\begin{gather*}
	\Prob{X > (1 + \delta)\mu_U} \leqslant \exp\left(- \frac{\delta^2}{3}\mu_U\right),
    \\
	\Prob{X < (1 - \delta)\mu_L} \leqslant \exp\left(- \frac{\delta^2}{2}\mu_L\right).
\end{gather*}
\end{theorem}

\begin{theorem}[Hoeffding's Inequality~\cite{dubhashi2009concentration}]\label{thm:hoeffding}
Let $X_{1}, \ldots, X_{n}$ be independent random variables such that $X_{i} \in [a_i, b_i]$. 
Let $X = \sum_{i=1}^n X_i$ and let $\mu = \Ex{X}$. 
Then, for all $t>0$,
\[
    \Prob{| X - \mu | \geq t} 
    \leq 2\exp\left(-\frac{2 t^{2}}{\sum_{i=1}^{n} (b_i - a_i)^2}\right).
\]
\end{theorem}


\bibliographystyle{alpha}
\bibliography{references}

\end{document}